\documentclass[5p,twocolumn,10pt]{elsarticle} %SPM format with Times New Roman font

\usepackage{amsmath}
\usepackage{hyperref}
\usepackage[toc,page]{appendix}

\usepackage{lineno}
\usepackage{hyperref}
\usepackage{enumitem}
\modulolinenumbers[5]

\usepackage{subcaption}
\usepackage{mathtools}
\usepackage{amsmath}
\usepackage{amssymb}
\usepackage{stmaryrd}
\usepackage{amsthm}
\usepackage{amsfonts}
\usepackage{dsfont}
\usepackage{graphicx}
\usepackage{upgreek}
\usepackage{bm}
\usepackage{color}
\usepackage{float}
\usepackage{tikz-cd}
\usepackage{scrextend}
\usepackage[cal=boondoxo]{mathalfa}
\usepackage{nicefrac}
\usepackage{algorithm,algpseudocode} 
\usepackage{todonotes}

\biboptions{numbers,sort&compress}

\DeclareFontFamily{OT1}{pzc}{}
\DeclareFontShape{OT1}{pzc}{m}{it}{<-> s * [1.200] pzcmi7t}{}
\DeclareMathAlphabet{\mathpzc}{OT1}{pzc}{m}{it}

\newcommand{\bx}{{\mathbf{x}}}

\newcommand{\indic}{{\mathbf{1}}}

\newcommand{\dimm}{\mathrm{d}}

\newcommand{\sweep}{\mathsf{sweep}}
\newcommand{\unsweep}{\mathsf{unsweep}}

\newcommand{\interior}{\mathsf{i}}
\newcommand{\exterior}{\mathsf{e}}

\newcommand{\R}{{\mathds{R}}}

\newcommand{\SE}[1]{{\mathrm{SE}(#1)}}

\newcommand{\shape}{S}

\newcommand{\motion}{M}

\newcommand{\field}{\mathcal{f}}
\newcommand{\gield}{\mathcal{g}}
\newcommand{\fs}{\mathsf{FS}}
\newcommand{\fv}{\mathsf{FV}}
\newcommand{\cell}{C}
\newcommand{\dirac}{{\bm\updelta}}
\newcommand{\weight}{w}
\newcommand{\TS}{{\mathcal{D}}}

\renewcommand{\th}{$^\text{th}$ }

\theoremstyle{definition}

\newtheorem{defn}{Definition}

\newtheorem{lemma}{Lemma}
\newtheorem{coro}{Corollary}

\newcommand{\eq}[1]{(\ref{#1})} %Use \eq{...} for referincing equations as (...)
\newcommand{\com}[1]{} %Use \com{...} to commenting out multiple paragraphs

\journal{Computer-Aided Design}

\begin{document}
	
\baselineskip11pt
\begin{frontmatter}

\title{Co-generation of Collision-Free Shapes for Arbitrary One-Parametric Motion}

\author{Clinton B. Morris and Morad Behandish}

\address{\rm
	Palo Alto Research Center (PARC),
	3333 Coyote Hill Road, Palo Alto, California 94304
	\vspace{-15.0pt}
}

\begin{abstract}

Mechanical assemblies can exhibit complex relative motions, during which collisions between moving parts and their surroundings must be avoided. To define feasible design spaces for each part's shape, ``maximal'' collision-free pointsets can be computed using configuration space modeling techniques such as Minkowski operations and sweep/unsweep. For example, for a pair of parts undergoing a given relative motion, to make the problem well-posed, the geometry of one part (chosen arbitrarily) must be fixed to compute the maximal shape of the other part by an unsweep operation. Making such arbitrary choices in a multi-component assembly can place unnecessary restrictions on the design space. A broader family of collision-free pairs of parts can be explored, if fixing the geometry of a component is not required. In this paper, we formalize this family of collision-free shapes and introduce a generic method for generating a broad subset of them. Our procedure, which is an extension of the unsweep, allows for co-generation of a pair of geometries which are modified incrementally and simultaneously to avoid collision. We demonstrate the effectiveness and scalability of our procedure in both 2D and 3D by generating a variety of collision-free shapes. Notably, we show that our approach can automatically generate freeform cam and follower profiles, gear teeth, and screw threads, starting from colliding blocks of materials, solely from a specification of relative motion and without the use of any feature-informed heuristics. Moreover, our approach provides continuous measures of collision that can be incorporated into standard gradient-descent design optimization, allowing for simultaneous collision-free and physics-informed co-design of mechanical parts for assembly.
\end{abstract}

\begin{keyword}
    Automated Design \sep
    Collision Avoidance \sep
	Configuration Space \sep
	Spatial Reasoning \sep
	Persistent Contact
\end{keyword}

\end{frontmatter}

%\linenumbers
\section{Introduction} \label{sec_intro}

Design for manufacturing and assembly (DfM/DfA) \cite{Boothroyd1996design}, popularized by Boothroyd and Dewhurst in the early 1980s, is a set of principles for design of manufacturable parts to enable faster, cheaper, and more reliable assemblies. These principles include minimizing the number of parts and assembly operations, making these operations simple and fail-proof, standardizing the interfaces, etc.
Such design rules and evaluation criteria can be effective as guidelines for manual design or fairly limited automation with parameterized geometric features that can be manufactured with traditional processes. With the advent of modern (e.g., additive and hybrid \cite{Chong2018review}) manufacturing, there are tremendous opportunities for DfM/DfA demanding a deeper look into geometric and physical modeling of assemblies.

Computational design tools such as gradient-based shape and topology optimization (SO/TO) \cite{Sokolowski1992introduction,Sigmund2013topology} have proven effective to satisfy physics-based performance criteria at the individual part level, owing to their formulation as differentiable objective functions and constraints. However, incorporating assembly and manufacturing constraints, often formulated using geometric and kinematic models, is challenging \cite{Mirzendehdel2019exploring}. A quantitative analysis of concepts such as collision, containment, contact, and complementarity of parts of arbitrarily shapes undergoing arbitrary motions is required, before they can be cast into differentiable measures. This paper focuses on collision and containment, quantified by overlap measures between shape indicator functions \cite{Lysenko2013fourier}, and sets the groundwork of an extensible framework to more complex concepts such as contact and complementarity, quantified by more complex shape functions \cite{Lysenko2016effective,Behandish2017shape}.

Consider, for example, the case of a single degree-of-freedom (DOF) mechanism such as a four-bar or slider-crank linkage, a cam/follower or pinion/gear pair, a latching mechanism, etc. The design of such assemblies often starts from a system-level design (e.g., determining types and relative positions of joints on each link) to satisfy kinematic and dynamic requirements, determining the relative motions between parts \cite{Norton2008design}. While these computations are automated and available in commercial packages, there is a gap in cascading them to part-scale 3D design with any level of generality. For example, there is no systematic way to impose collision avoidance between the crank and coupler in a four-bar linkage or persistent contact between the cam and follower with prescribed motion, while optimizing their shapes for performance (e.g., stiffness and strength) subject to manufacturing constraints \cite{Mirzendehdel2020topology,Mirzendehdel2021optimizing,Mirzendehdel2022topology}. Fig. \ref{fig_engine} \cite{Chakravarthi2017} illustrates a more complex assembly with various collision, containment, contact, and complementarity requirements that are difficult to formalize for design automation. Ultimately, to automate the design or redesign of such a complex assembly, all of these concepts must be integrated into automated design workflows. While designing such a complex assembly is beyond current capabilities, automated collision avoidance is a necessary advancement to achieve this goal. 

\begin{figure}[ht]
	\centering
	\includegraphics[width=0.25\textwidth]{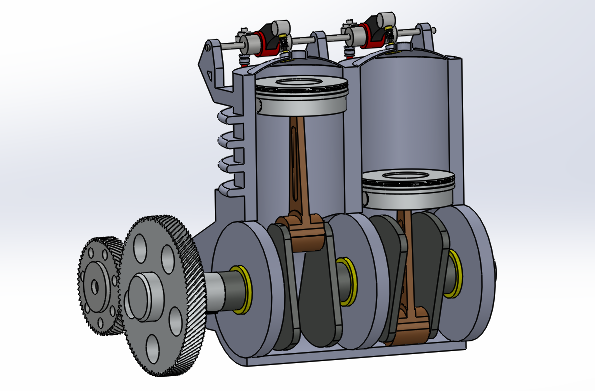}
	\caption{Cross-section of engine demonstrating collision, contact, containment, and complementarity constraints to properly function \cite{Chakravarthi2017}.}    
	\label{fig_engine}
\end{figure}

In this paper, we devote specific attention to collision avoidance or (equivalently) containment constraints, as they capture a broad class of design criteria involving interactions of shapes and motions \cite{Ilies2000shaping,Ilies2002class}. To incorporate such constraints into gradient-based optimization, the extent of their violations must be locally evaluable and differentiable with respect to design variables (e.g., small geometric or topological changes to the design), to penalize the objective functions. Additionally, collision is a pairwise relation, presenting additional challenges when multiple moving parts are to be designed simultaneously. Changes to one part will potentially introduce or eliminate collisions with the others, necessitating a {\it co-design} workflow in which the parts are shaped simultaneously, evolution of one part directly informing those of the others in the assembly. %We define this 

To the best of our knowledge, such a co-design workflow that enables designing complex shapes under arbitrary motions does not exist. This paper presents a framework to develop such workflows to design {\it families} of parts in which collision avoidance constraints can be simultaneously satisfied and seamlessly integrated with other (e.g., performance and manufacturing) constraints. More specifically, we present a procedure to generate members of a family of ``maximal pairs'' of collision-free parts. We do not divert our attention to formulating other constraints in this paper, although we show a general recipe to couple them with collision constraints.

\subsection{Related Work} \label{sec_lit}

Existing approaches to generating shapes that satisfy motion-based collision and contact constraints place excessive restrictions on the design space, missing opportunities for better-performing and more cost-effective assemblies. Ilie\c{s} and Shapiro \cite{Ilies2000shaping,Ilies2002class} developed a framework to produce a ``maximal'' shape for a part, moving against another part of prescribed shape, to satisfy collision avoidance or (equivalently) containment constraints. This led to the definition of a fundamental new solid modeling operation called unsweep \cite{Ilies1999dual,Ilies1997unsweep}. Nelaturi and Shapiro \cite{Nelaturi2011configuration} extended the idea to a broader class of configuration space operations, based on group morphology \cite{Lysenko2010group}. These operations have proven effective in solving various manufacturing analysis and process planning problems \cite{Behandish2018automated,Behandish2019classification,Nelaturi2019automatic}. 

There are at least two challenges with using such operations, formulated in a set-theoretic language, in a design framework. First, using them to compute a maximal entity (shape or motion) requires full knowledge of all of the other entities---for instance, sweep/unsweep map a given shape and a given motion to a maximal shape, while homogeneous Minkowski products/quotients map a given pair of shapes to a maximal motion, avoiding collisions \cite{Lysenko2010group}. Second, these set-theoretic operations do not interoperate well with other (e.g., performance and manufacturing) constraints, commonly expressed using analytic inequalities. Unlike the latter, the former do not provide quantitative measures of violation of collision constraints, which are critical for navigating {\it tradeoffs} with other constraints. Measure-theoretic generalizations of such operations \cite{Behandish2017analytic} based on convolution algebras \cite{Lysenko2010group} can be used to obtain locally evaluable and differentiable measures that can be integrated with other constraints via Lagrange multipliers \cite{Mirzendehdel2019exploring}. However, the attention has so far been restricted to single part design.

In contrast, our procedure does not require either part's shape to be fixed upfront and can generate families of maximal collision-free designs (in a partial ordering of pairs), using local and global measures that can be integrated with other constraints and differentiated for gradient-descent optimization. We present an incremental co-generation procedure that subsumes one-way unsweeps of either part, against the initial design of the other part's complement, as extreme cases of the family of maximal pairs. 

To avoid premature decisions and excessive geometric constraints on part geometries, St\"ockli and Shea \cite{Stockli2020topology} developed an automated procedure for simultaneously modifying the shapes of two parts in relative motion. Co-generation of the shapes was achieved using a TO procedure in which they formulated and minimized an aggregate collision measure. They showed that such measures can be rapidly computed through matrix multiplications during the optimization loop to globally quantify collisions. To update the shapes within an iterative optimization loop to alleviate the collisions, basic rules were applied to grow and shrink the shapes.

Our computational algorithm has similar elements to \cite{Stockli2020topology} in the way collision measures are computed using a voxelization of the moving domains, although we present a formal set-theoretic groundwork, built on top of \cite{Ilies2000shaping,Ilies2002class,Ilies1997unsweep,Nelaturi2011configuration} that is not tied to a specific representation scheme and can be discretized in many different ways. Moreover, we use gradients and local measures of the collision measures to augment the TO sensitivity fields, enabling further scalability. Most importantly, our approach generates families of maximal shapes as opposed to a single arbitrary pair obtained using rule-based heuristics in \cite{Stockli2020topology}.

\subsection{Contributions \& Outline}

This article presents a general formulation and computational framework for co-generating maximal collision-free shapes in arbitrary relative motion. We show that:

\begin{enumerate}
    \item The unsweep operation and the underlying notion of {\it maximality} (for a partial ordering of solids) can be expanded to pairs of solids to more broadly explore the space of collision-free designs (Section \ref{sec_prelim}).
    \item The collision of solids in relative motion can be measured locally and globally to use as a differentiable violation measure to penalize gradient-based optimization (Section \ref{sec_collisionMeasure}).
    \item The violation measure can be used to formulate an iterative and incremental co-generation procedure for a broad subset of maximal pairs of collision-free solids (Section \ref{sec_method}).
    \item The procedure can scale to efficiently co-generate a variety of nontrivial collision-free solids that, if allowed by motions, exhibit persistent contact\footnote{Persistent contact means nonempty boundary intersection maintained throughout the motion. Here, contact is not explicitly enforced, although it could be (e.g., by adding contact or complementarity measures to the objective function). Its persistence cannot be guaranteed by shape design, as it also depends on the given motion. When it is possible, maximizing volume subject to collision-avoidance appears to improve its chances. Understanding persistent contact properties of shapes and motions requires further research (out of scope here).}  in 2D and 3D (Section \ref{sec_results}).
    \item The violation measures enable using collision-avoidance constraints in gradient-based optimization to extend beyond geometric reasoning (Section \ref{sec_conclusion}).
\end{enumerate}
We also show that the collision measure can be efficiently computed by linear-algebraic operations, involving a precomputed pairwise correlation matrix that depends on the motion, pre- and post-multiplied by arrays of design variables, all of which can be parallelized on the CPU/GPU (Section \ref{sec_collisionMeasure}). However, this separability of computations was first observed by St\"ockli and Shae \cite{Stockli2020topology}, hence does not constitute a novel contribution of this paper.

\section{Preliminaries}  \label{sec_prelim}

Kinematic design of mechanical parts and assemblies under prescribed motions can often be formulated as computing ``maximal'' pointsets that satisfy collision avoidance or (equivalently) containment constraints. The maximality is defined in the partial ordering of pointsets with respect to containment. Such maximal shapes prune the design space to a feasible subspace, defined by the powerset of the maximal pointset \cite{Mirzendehdel2019exploring}, for subsequent design space exploration to satisfy various performance and manufacturing criteria (e.g. via TO \cite{Mirzendehdel2020topology,Mirzendehdel2019exploring,Mirzendehdel2021optimizing,Mirzendehdel2022topology,Iyer2021pato}).

Let $\shape \subset \R^\dimm$ be a $\dimm-$dimensional solid or `r-set', defined as a compact (bounded and closed) regular and semianalytic pointset in the Euclidean $\dimm-$space $\R^\dimm$ \cite{Requicha1980representations}. Let $\motion \subset \SE{\dimm}$ be a motion acting on $\R^\dimm$, i.e., a parameterized collection of rigid configurations or `poses' (combined rotations and translations) that a $\dimm-$dimensional pointset can assume \cite{Lozano-Perez1990spatial}. %(Fig. \ref{?})%
For computational purposes, the solids are commonly represented by boundary or volumetric representations (e.g., NURBS, CSG, surface/volume mesh, voxels, or sampled point clouds) \cite{Requicha1980representations}, while motions are typically represented by sampled or parameterized homogeneous matrices, dual quaternions, etc. We restrict our attention to $\dimm := 2, 3$ and one-parametric motions.

The first class of problems can be formulated and solved using the sweep and unsweep operations \cite{Ilies1997unsweep}. The sweep of a given solid $\shape$ under a given motion $\motion$ is another pointset $\sweep(\motion, \shape) \subset \R^\dimm$, which is a superset of $\shape$ collecting all points that are included in the displaced shape $\tau \shape$ for at least one configurations $\tau \in \motion$:
\begin{itemize}
    \item Explicit definition (by an indexed union):
    \begin{equation*}
        \sweep(\motion, \shape) \triangleq \bigcup_{\tau \in \motion} \tau \shape.
    \end{equation*}
    \item Implicit definition (by a membership test):
    \begin{align*}
        \sweep(\motion, \shape) &\triangleq \big\{ \bx \in \R^\dimm ~|~ \exists \tau \in \motion : \bx \in \tau \shape \big\} \\
        &= \big\{ \bx \in \R^\dimm ~|~ \motion^{-1} \bx \cap \shape \neq \emptyset \big\},
    \end{align*}
    where $\motion \bx \triangleq \{\tau \bx~|~ \tau \in \motion\}$ and $\tau \shape \triangleq \{\tau \bx ~|~ \bx \in \shape\}$ where $\tau \bx \in \R^\dimm$ denotes a displaced point.
\end{itemize}
Note that the membership of a query point $\bx \in \R^\dimm$ in the sweep is tested by applying the inverse motion to the query point and checking whether the resulting trajectory passes through the given shape.

On the other hand, the unsweep of a given solid $\shape$ under a given motion $\motion$ is another pointset $\unsweep(\motion, \shape) \subset \R^\dimm$, defined by the subset of $\shape$ including all points that remain included in the displaced shape $\tau \shape$ for all configurations $\tau \in \motion$:
\begin{itemize}
    \item Explicit definition (by an indexed intersection):
    \begin{equation*}
        \unsweep(\motion, \shape) \triangleq \!\! \bigcap_{\tau \in \motion^{-1}} \!\! \tau \shape.
    \end{equation*}
    \item Implicit definition (by a membership test):
    \begin{align*}
        \unsweep(\motion, \shape) &\triangleq \big\{ \bx \in \R^\dimm ~|~ \forall \tau \in \motion^{-1} : \bx \in \tau \shape \big\} \\
        &= \big\{ \bx \in \R^\dimm ~|~ \motion \bx \cap \overline{\shape} = \emptyset \big\},
    \end{align*}
    where $\overline{\shape} \triangleq (\R^\dimm - \shape)$ is the set complement of $\shape$.
\end{itemize}
Note that the membership of a query point $\bx \in \R^\dimm$ in the unsweep is tested by applying the forward motion to the query point and checking whether the resulting trajectory remains inside the given shape.

To see how sweep and unsweep can be used in practical design problems, consider a pair of solids $\shape_1, \shape_2 \subset \R^\dimm$ moving under one-parametric motions $\motion_1, \motion_2 \subset \SE{\dimm}$. The solids can be parts in a single DOF mechanism, e.g., any pair of links in a four-bar linkage, a cam/follower pair, a pinion/gear pair, a latch/pin pair, etc. The relative motion (of $\shape_2$ as observed from a frame attached to $\shape_1$) is $\motion \triangleq \motion_1^{-1}\motion_2$. If we fix the shape of solid $\shape_1$, the maximal pointset that does not collide with $\shape_1$ (i.e., is contained within $\overline{\shape}_1$) throughout the motion is:
\begin{equation}
    \shape_2^\ast \triangleq \overline{\sweep(\motion^{-1}, \shape_1)} = \unsweep(\motion, \overline{\shape}_1).
\end{equation}
The maximality implies that every other solid that satisfies the no-collision constraint will be a subset of the maximal pointset (i.e., $(\motion \shape_2 \cap \shape_1) = \emptyset$ iff $\shape_2 \subseteq \shape_2^\ast$).

Let us assume the motion can be parameterized as $\motion = \{ \tau(t) ~|~ 0 \leq t \leq 1\}$. For the $1-$DOF mechanisms exemplified above, $t \in [0, 1]$ can be thought of as normalized time---or any parameter monotonically changing with time---over the mechanism's motion cycle. Consider the car hood latch example in Fig. \ref{fig_latch}, revisited from \cite{Ilies2004equivalence}. 

\begin{figure}[h]
    \centering
    \includegraphics[width=0.4\textwidth]{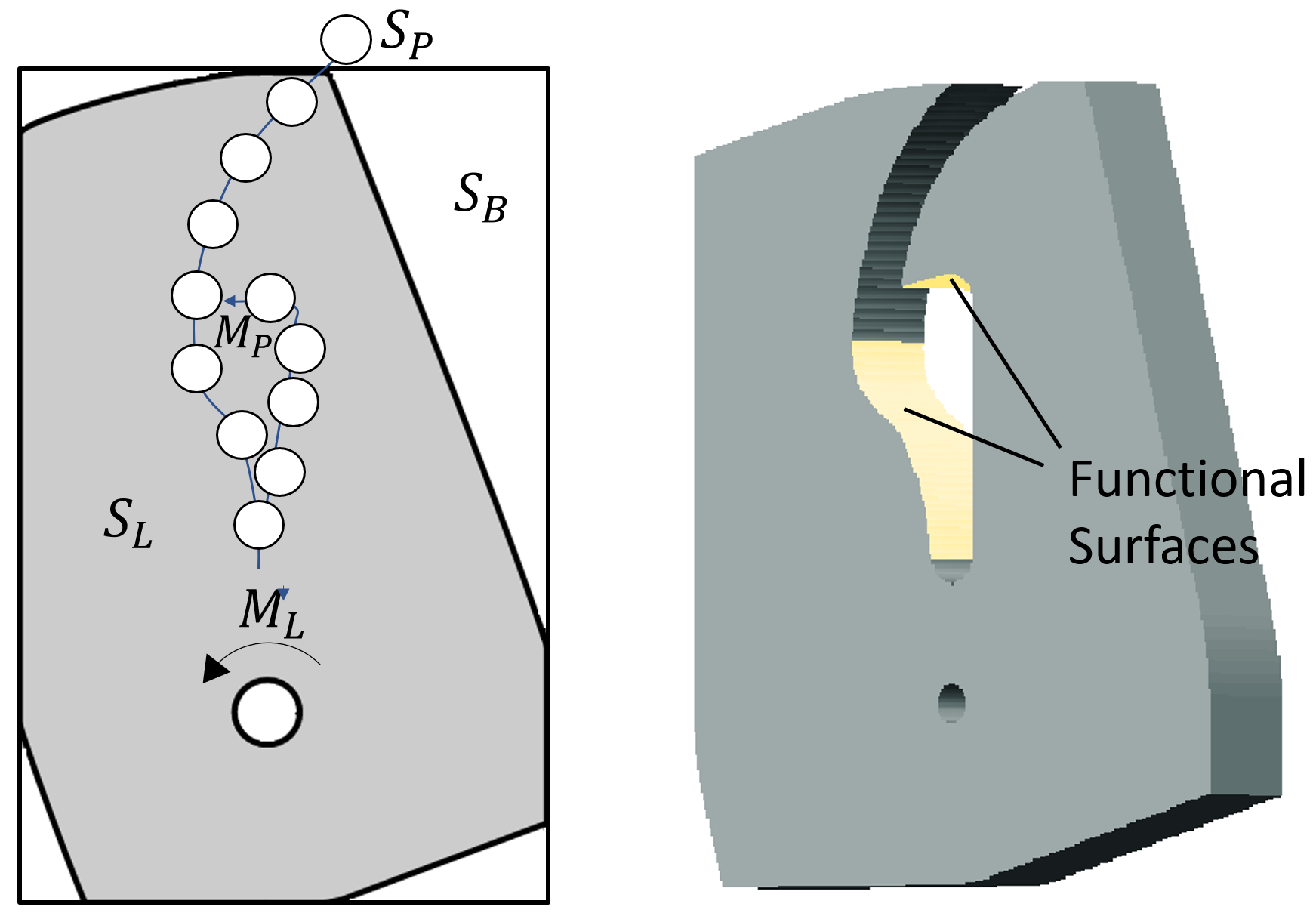}
    \caption{Car hood latch example adopted from \cite{Ilies2004equivalence}, in which a rotating latch must remain contained in an enclosure while maintaining contact (without collision) over a functional surface with a pin translating vertically. The relative motion is shown on the left.}    
    \label{fig_latch}
\end{figure}
\begin{itemize}
    \item Our goal is to design the shape of a latch $\shape_\text{L}$, rotating around a pivot by a given trajectory $\motion_\text{L}$ with respect to the car (common frame of reference).
    \item {\bf Constraint 1:} Let $\shape_\text{B}$ represent a bounding region that the latch must not exit to avoid collision with surrounding objects under the hood that are stationary with respect to the car (hence $\motion_\text{B}$ is identity). 
    \item {\bf Constraint 2:}  Let $\shape_\text{P}$ represent a pin of given shape, attached to the hood, thus moving along a given trajectory $\motion_\text{P}$ with respect to the car. The latch must avoid collision while maintaining contact with the pin, as it moves with a relative motion $\motion_\text{P}^{-1} \motion_\text{L}$.
\end{itemize}
The design space for the latch $\shape_\text{L}$ is thus restricted (by the conjunction of these constraints) to subsets of
\begin{equation}
   % \shape_\text{L}^\ast := \unsweep(\motion_\text{L}, \shape_\text{L}) \cap \unsweep(\motion_\text{P}^{-1} \motion_\text{L}, \shape_\text{P}).
    \shape_\text{L}^\ast := \unsweep(\motion_\text{B}^{-1} \motion_\text{L}, \shape_\text{B}) \cap \unsweep(\motion_\text{P}^{-1} \motion_\text{L}, \shape_\text{P}).
\end{equation}
A similar approach can be used to design the cam profiles for given follower profiles and cam/follower motions. Despite its generality (for arbitrarily complex shapes and motions) and rigor, this approach has three major limitations that we aim to address:

First, using unsweep requires pre-defining one of the two shapes in relative motion to obtain a maximal shape for the other one. This requirement may be fine for design problems with clear choices for the shapes of certain components or spatial regions because of non-negotiable conditions (e.g., the bounding box in Fig. \ref{fig_latch}) or simple profiles (e.g., the cylindrical pin in Fig. \ref{fig_latch} and cylindrical/flat followers). However, in problems where both shapes in relative motion are to be designed simultaneously (e.g., the cam/follower system), the arbitrary choice of either shape limits the design freedom and requires trial-and-error iterations over a theoretically uncountable---and computationally exponential---number of possibilities. 

Second, the maximal pointsets that satisfy these constraints come with no guarantees for effective contact, shape complementary, force/torque transmission, and other requirements, whose quantification for arbitrary shapes is nontrivial. For example, consider two parts in a relative {\it screw} motion, each to be selected as a subset of the bounding boxes. If one of the parts is shaped as a cylindrical hole or a threaded nut, the maximal shape for the other part will be computed by an unsweep as a matching cylindrical peg or a threaded bolt, respectively. The contact and complementary is perfect in this case, because the screw motions form a symmetry subgroup of $\SE{3}$ and the first part's shape (the nut) was carefully selected such that this particular screw motion of the same pitch is its symmetry; hence unsweep generates a matching shape (the bolt) that, not only avoids collision with the nut, but is also a perfect complement of it and remains as such throughout the motion.%
\footnote{To be accurate, these statements apply strictly to an infinite extension of a bolt/nut of the same pitch, meaning that applying the motion to their shapes does not change the shapes (symmetry) and they partition the $3-$space (complementarity) throughout the motion.}
However, if we picked a slightly different shape for the first part, such as a half-cylindrical hole or a nut of a different pitch, the unsweep would generate a maximal shape that, despite avoiding collision, offers poor contact and complementarity. This example illustrates why co-design is a more practical approach where both shapes are chosen and refined simultaneously. 

The third limitation is that set-theoretic operations such as unsweep, used to express kinematic constraints, do not interoperate well with analytic constraints such as global or local inequalities in terms of field variables, used to express physics-based constraints (e.g., stiffness or strength). As a result, concurrently satisfying kinematic and physical constraints is challenging \cite{Nelaturi2019automatic}. One possible approach is to translate set-theoretic operations with measure-theoretic parallels, e.g., by replacing Boolean and morphological operations of pointsets with pointwise logical operations and convolutions of their indicator functions, respectively \cite{Behandish2017analytic}. The later extend the former by not only providing binary membership tests for set operations, but also continuous and differentiable membership ``grades'' that measure violations of constraints and can be penalized to sensitivity fields for TO \cite{Mirzendehdel2019exploring,Mirzendehdel2020topology}. 

We overcome the first limitation by defining an analytic generalization of unsweep to enable incremental co-generation of families of ``maximal pairs'' that satisfy the no-collision or (equivalently) containment constraints. We apply a simple design rule (namely, volume maximization) subject to these constraints, to encourage persistent contact, i.e., contact maintained throughout the motion. While this partially overcomes the second limitation stated above, more sophisticated contact or complementarity measures \cite{Lysenko2016effective,Behandish2017shape} can be used in the objective function to enforce persistent contact. It is important to note that persistent contact depends not only on shape design, but also on the properties of motion (fixed here) which are not sufficiently understood. We do not address these issues in this paper to maintain focus on the co-generation process to incrementally minimize and eliminate collision. Our approach in this paper is in the same spirit of \cite{Behandish2017analytic}, providing a measure-theoretic formulation of sweep/unsweep and collision that is differentiable with respect to the field variables and can thus overcome the third limitation stated above.

\section{Collision Measures} \label{sec_collisionMeasure} \label{sec_col}

In this section, we present a formal modeling framework (Section \ref{sec_set}) to define maximal collision-free pairs of solids under given one-parametric motions, and a discretization scheme (Section \ref{sec_disc}) to make them computable. In particular, we show that the computations can be factored into offline pre-computation of a motion-dependent and shape-independent correlation matrix and online computation of collision measures for specific shapes (i.e., updated during iterative design) via fast matrix multiplication.
  
\subsection{Set-Theoretic Formulation} \label{sec_set}
  
Let us consider two design domains $\Omega_1, \Omega_2 \subseteq \R^\dimm$ for the two solid we aim to co-design ($\dimm = 2, 3$). Let us begin from two given initial designs $\shape_1 \subseteq \Omega_1$ and $\shape_2 \subseteq \Omega_2$ representing the shapes of the two solids at rest (i.e., before applying the motion). Let $\motion_1, \motion_2 \in \SE{\dimm}$ represent the motions that these solids experience with respect to a common frame of reference. We restrict our attention to one-parametric motions defined by the following sets:
\begin{align}
    \motion_1 \triangleq \big\{ \tau_1(t) ~|~ 0 \leq \tau \leq 1 \big\}, \\
    \motion_2 \triangleq \big\{ \tau_2(t) ~|~ 0 \leq \tau \leq 1 \big\},
\end{align}
where $\tau_1, \tau_2 : [0, 1] \to \SE{\dimm}$ are continuously time-variant configurations, and can be represented by homogeneous matrices, vector-quaternion pairs, dual quaternions, etc. The displaced solids at any given time $t \in [0, 1]$ are:
\begin{align}
    \shape_1(t) \triangleq \tau_1(t)\shape_1 = \big\{ \tau_1(t) \bx ~|~ \bx \in \shape_1 \big\}, \\
    \shape_2(t) \triangleq \tau_2(t)\shape_2 = \big\{ \tau_2(t) \bx ~|~ \bx \in \shape_2 \big\}.
\end{align}
Without loss of generality, we assume $\tau_1(0) = \tau_2(0)$ to be identity, so $\shape_1(0) = \shape_1$ and $\shape_2(0) = \shape_2$, as intended.

To formulate collision measures, it is more convenient to represent the two pointsets implicitly via binary membership tests, also called indicator or characteristic functions $\indic_{\shape_1}, \indic_{\shape_2} : \R^3 \to \{0, 1\}$, defined generally by:
\begin{equation}
    \indic_{\shape}(\bx) \triangleq \left\{
    \begin{array}{ll}
         1 & \text{if}~ \bx \in \shape,  \\
         0 & \text{otherwise}.
    \end{array}
    \right.
\end{equation}
Note that indicator functions are contra-variant with rigid transformations, i.e., $\indic_{\tau \shape} (\bx) = \indic_{\shape}(\tau^{-1} \bx)$, meaning that a membership query for a given point against the displaced solid can be computed by displacing the query point along the inverse trajectory and testing its membership against the stationary solid.

Let $\motion = \motion_1^{-1} \motion_2$ stand for the relative motion of $\shape_2$ as observed from a frame of reference attached to $\shape_1$, noting that by kinematic inversion, $\motion^{-1} = \motion_2^{-1} \motion_1$ would represent the relative motion of $\shape_1$ as observed from a frame of reference attached to $\shape_2$.

Let $\motion_{1,2}$ (resp. $\motion_{2,1}$) represent the relative motions of $\shape_2$ (resp. $\shape_1$) as observed from a moving frame of reference attached to $\shape_1$ (resp. $\shape_2$):
\begin{align}
    \motion_{1,2} \triangleq \motion_1^{-1} \motion_2 = \big\{\tau_{1,2}(t) ~|~ 0 \leq t \leq 1 \big\}, \\
    \motion_{2,1} \triangleq \motion_2^{-1} \motion_1 = \big\{\tau_{2,1}(t) ~|~ 0 \leq t \leq 1 \big\},
\end{align}
where $\tau_{1,2}(t) = \tau_1^{-1}(t) \tau_2(t)$ and $\tau_{2,1}(t) = \tau_2^{-1}(t) \tau_1(t)$. Note also that $\tau_{1,2}(t) = \tau_{2,1}^{-1}(t)$ hence $\motion_{1,2}(t) = \motion_{2,1}^{-1}(t)$.

The displaced solids at any given time $t \in [0, 1]$ in the relative frames are:
\begin{align}
    \shape_{1,2}(t) \triangleq \tau_{2,1}(t)\shape_1 = \big\{ \tau_{2,1}(t) \bx ~|~ \bx \in \shape_1 \big\}, \label{eq_S12} \\
    \shape_{2,1}(t) \triangleq \tau_{1,2}(t)\shape_2 = \big\{ \tau_{1,2}(t) \bx ~|~ \bx \in \shape_2 \big\}. \label{eq_S21}
\end{align}

To quantify the contribution of every point $\bx \in \R^3$, attached to one shape, to its collision with the other shape, we measure the duration over which the point's trajectory collides with the latter:
\begin{align}
    \field_{2,1}(\bx) &\triangleq \int_0^1 \indic_{\shape_{2,1}(t)}(\bx) ~dt = \int_0^1 \indic_{\shape_2}(\tau_{2,1}\bx) ~dt, \label{eq_f1_} \\
    \field_{1,2}(\bx)  &\triangleq \int_0^1 \indic_{\shape_{1,2}(t)}(\bx) ~dt = \int_0^1 \indic_{\shape_1}(\tau_{1,2}\bx) ~dt, \label{eq_f2_}
\end{align}
as illustrated in Fig. \ref{fig_shapes}. To eliminate the contribution of the points that are outside each shape, we can multiply by the indicator functions of each shape:
\begin{align}
    \overline{\field}_{2,1}(\bx) &\triangleq  \int_0^1 \indic_{\shape_2}(\tau_{2,1}\bx) \indic_{\shape_1}(\bx) ~dt,  \label{eq_f1} \\
    \overline{\field}_{1,2}(\bx) &\triangleq  \int_0^1 \indic_{\shape_1}(\tau_{1,2}\bx) \indic_{\shape_2}(\bx) ~dt.  \label{eq_f2}
\end{align}

\begin{figure}[ht]
   	\centering
    \includegraphics[width=0.45\textwidth]{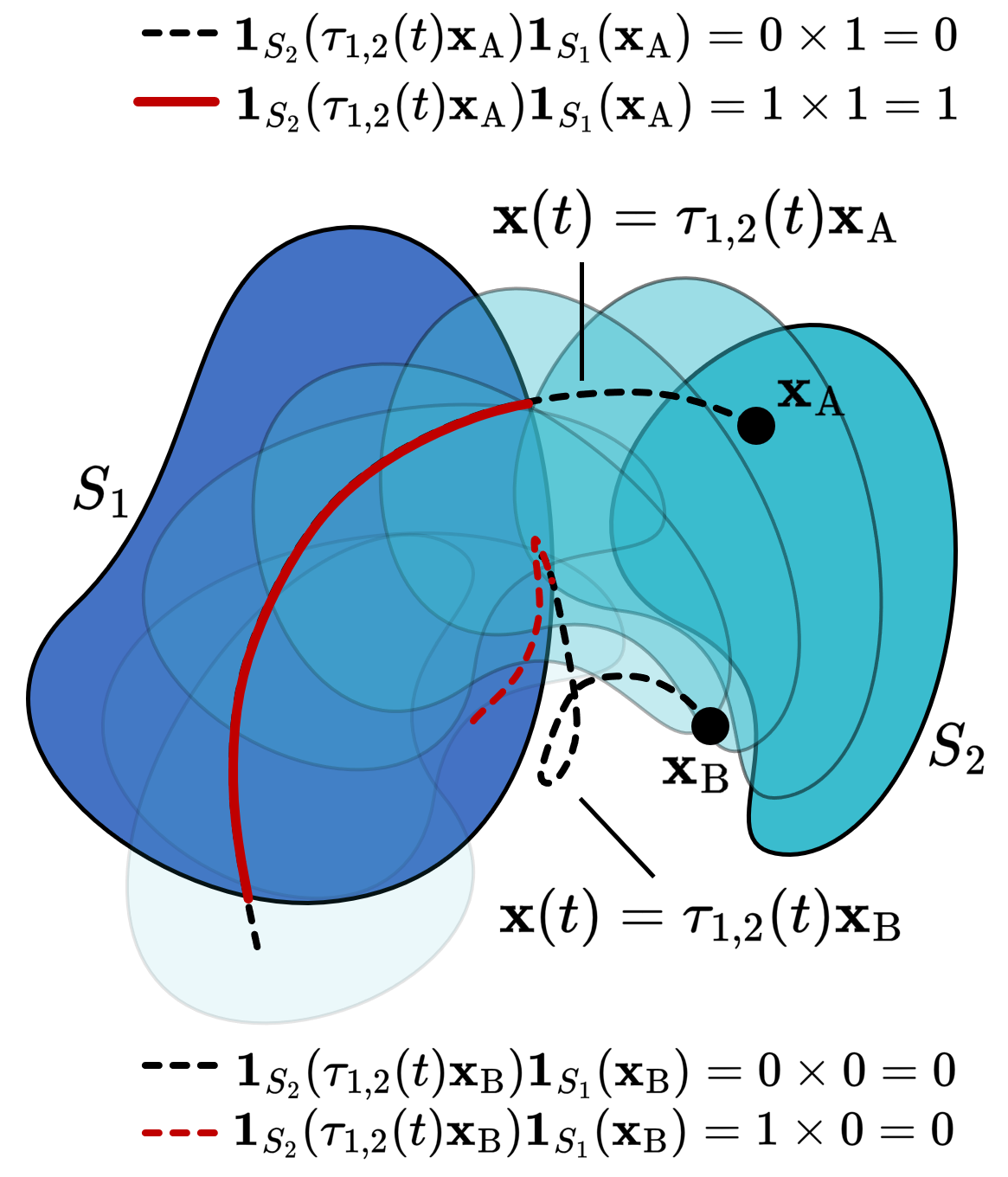}
    \caption{Local collision measure $\overline{\field}_{2,1}(\bx)$ of two query points $\mathbf{x_1}$ and $\bx \in \{ \bx_\mathrm{A}, \bx_\mathrm{B} \}$. The relative trajectory of each query points is determined by the relative configurations, $\tau_{1,2}(t)$ of $S_{2}$ to $S_{1}$.} \label{fig_shapes}
\end{figure}   

\begin{lemma}
    The two functions in \eq{eq_f1} and \eq{eq_f2} are implicit representations of the following unsweeps as their $0-$superlevel sets (up to regularization), i.e., the following statement is true ``almost everywhere'':%
    \footnote{A statement is true ''almost everywhere'' for pointsets with respect to a $d-$measure if it be untrue only over $\dimm-1$ or lower-dimensional pointsets, which disappear upon topological regularization \cite{Behandish2017analytic}.}
    \begin{align}
        \bx \in \unsweep(\motion_{1,2}, \overline{\shape}_2) \cap \Omega_1 ~\overset{\mathsf{ae}}{\rightleftharpoons}~ \overline{\field}_{2,1}(\bx) = 0, \\
        \bx \in \unsweep(\motion_{2,1}, \overline{\shape}_1) \cap \Omega_2 ~\overset{\mathsf{ae}}{\rightleftharpoons}~ \overline{\field}_{1,2}(\bx) = 0,
    \end{align}
\end{lemma}

\begin{proof}
    The detailed derivation is given in \cite{Behandish2017analytic}. The idea is that by definition, points in a given $\unsweep(\motion, \overline{\shape})$ are defined by containment in $\overline{\shape}$ (i.e., no intersection with $\shape$) under the motion $\motion$. Asserting zero collision measures is a similar criteria, however, the $\dimm-$integrals cannot capture the $\dimm-1$ and lower-dimensional interference (i.e., contact), hence the equivalence is almost everywhere.
\end{proof}

To derive global measures (a single value for each solid) from the above local measures, we can integrate them over the respective solids:
\begin{align}
    \gield^{}_{2,1} &\triangleq \int_{\shape_1} \field_{2,1}(\bx) ~d\mu^\dimm[\bx] = \int_{\Omega_1} \overline{\field}_{2,1}(\bx) ~d\mu^\dimm[\bx], \label{eq_g1} \\
    \gield^{}_{1,2} &\triangleq \int_{\shape_2} \field_{1,2}(\bx) ~d\mu^\dimm[\bx] = \int_{\Omega_2} \overline{\field}_{1,2}(\bx) ~d\mu^\dimm[\bx], \label{eq_g2}
\end{align}
where $d\mu^\dimm[\cdot]$ stands for differential $\dimm-$measure, i.e, the area/volume of an infinitesimal 2D/3D region at $\bx \in \R^\dimm$. The goal of collision-free co-design is to find a ``maximal pair'' of solids $\shape_1 \subseteq \Omega_1$ and $\shape_2 \subseteq \Omega_2$, in a sense that we shall define precisely below, such that $\gield^{}_{\shape_2} = \gield^{}_{\shape_2} = 0$, which is true iff for all $\bx \in \R^\dimm$, $\overline{\field}_{2,1}(\bx) = \overline{\field}_{1,2}(\bx) = 0$.

\begin{defn}
    A pair of solids $(\shape_1, \shape_2)$ are called collision-free under the relative motion $\motion_{1,2} = \motion_{2,1}^{-1}$ if:
    \begin{equation}
        \mu^{\dimm}[\shape_1 \cap \motion_{1,2}\shape_2] = 0 ~\rightleftharpoons~ \mu^{\dimm}[\motion_{2,1}\shape_1 \cap \shape_2] = 0.
    \end{equation}
\end{defn}

\begin{lemma}
    A pair of solids $(\shape_1, \shape_2)$ are collision-free iff the collision measures are zero, i.e., $\gield^{}_{2,1} = 0 ~\rightleftharpoons~ \gield^{}_{1,2} = 0$.
\end{lemma}

\begin{proof}
    The above collision criteria implies interference over a $\dimm-$dimensional region---while contact over lower ($\dimm-1$ and below) dimensional regions is allowed---over a $1-$measurable period of time. Substituting the local measures in \eq{eq_f1_} and \eq{eq_f2_} into the global measures in \eq{eq_g1} and \eq{eq_g2} and swapping the space and time integrals yield:
    \begin{align}
        \gield^{}_{2,1} &= \int_0^1 \int_{\Omega_1} \indic_{\shape_{2,1}(t)}(\bx) \indic_{\shape_1}(\bx) ~d\mu^\dimm[\bx]~dt, \label{eq_subs_g1} \\
        \gield^{}_{1,2} &= \int_0^1 \int_{\Omega_2} \indic_{\shape_{1,2}(t)}(\bx) \indic_{\shape_2}(\bx) ~d\mu^\dimm[\bx]~dt. \label{eq_subs_g2} 
    \end{align}
    At a given snapshot $t \in [0, 1]$, the inner integrals are nonzero iff $\mu^\dimm[\shape_1(t) \cap \shape_2(t)] = 0$ which, by switching the frame of reference to the ones attached to the respective solids, is equivalent to $\mu^\dimm[\shape_1 \cap \shape_{2,1}(t)] = 0$ and $\mu^\dimm[\shape_{1,2}(t) \cap \shape_2] = 0$. The outer integrals are nonzero if such a nonzero $\dimm-$measurable interference persists for a $1-$measurable period of time. 
\end{proof}

The global collision measures of \eq{eq_subs_g1} and \eq{eq_subs_g2} are differentiable with respect to small perturbations of the solids. At a first glance, this might be counter-intuitive as collisions are understood to be ``discrete'' events. However, the measure of collision is continuous with respect to small local changes in the geometry (e.g., infinitesimal topological inclusions), for a given relative motion, as discussed in \ref{app_diff_measure} in greater detail.

\begin{defn}
    A pair of collision-free solids $(\shape_1^\ast, \shape_2^\ast)$ under the relative motion $\motion_{1,2} = \motion_{2,1}^{-1}$ is ``maximal'' if any $\dimm-$measurable growth of either solid confined to the design domain (i.e., $\shape_1^\ast \subseteq \Omega_1$ and $\shape_2^\ast \subseteq \Omega_2$) while keeping the other solid unchanged, makes them no longer collision-free under the same motion.
    
    More precisely, for every other pair of solids $(\shape_1, \shape_2)$:
    \begin{itemize}
        \item if $\shape_1^\ast \subset \shape_1 \subseteq \Omega_1$ then $(\shape_1, \shape_2^\ast)$ is not collision-free; and
        \item if $\shape_2^\ast \subset \shape_2 \subseteq \Omega_2$ then $(\shape_1^\ast, \shape_2)$ is not collision-free.
    \end{itemize}
\end{defn}
The maximal pairs can be more formally defined by imposing a partial order relation $\preccurlyeq$ over the space of all collision-free pairs, where $(\shape_1, \shape_2) \preccurlyeq (\shape_1', \shape_2')$ iff $\shape_1 \subseteq \shape_1'$ and $\shape_2 = \shape_2'$ or $\shape_1 = \shape_1'$ and $\shape_2 \subseteq \shape_2'$. A pair $(\shape_1^\ast, \shape_2^\ast)$ is maximal if $(\shape_1^\ast, \shape_2^\ast) \preccurlyeq (\shape_1', \shape_2')$ implies $(\shape_1^\ast, \shape_2^\ast) = (\shape_1', \shape_2')$. 

Notice that for a maximal pair of collision-free solids $(\shape_1^\ast, \shape_2^\ast)$, if one solid is grown by the slightest bit, the only way to keep the pair collision-free for the same motion scenario is to shrink the other solid. This observation provides us with an incremental shape modification strategy to traverse a family of maximal pairs, akin to walking along a Pareto front of maximality. If we picture the design space of all pairs of solids within the domains $\Omega_1, \Omega_2 \subseteq \R^\dimm$, the maximal front is a higher-dimensional manifold that bounds the feasible design subspace.

\begin{coro}
    Given a pair of ``initial designs'' $\shape_1 \subseteq \Omega_1$ and $\shape_2 \subseteq \Omega_2$ that collide under one-parametric motions $\motion_1, \motion_2 \subset \SE{3}$, the following two pairs of maximally collision-free solids can be constructed via unsweep:
    \begin{align}
        (\shape_1, \shape^\ast_2), \quad\text{where}~ \shape^\ast_2 \triangleq \unsweep(\motion_{1,2}, \overline{\shape}_1) \cap \Omega_2, \label{eq_unsweep_1} \\
        (\shape^\ast_1, \shape_2), \quad\text{where}~ \shape^\ast_1 \triangleq \unsweep(\motion_{2,1}, \overline{\shape}_2) \cap \Omega_1, \label{eq_unsweep_2}
    \end{align}
    Note that these are extreme cases:
    \begin{itemize}
        \item $\shape_2^\ast$ is the maximal solid that does not collide with $\shape_1$, assuming the shape of $\shape_1$ is fixed.
        \item $\shape_1^\ast$ is the maximal solid that does not collide with $\shape_2$, assuming the shape of $\shape_2$ is fixed.
    \end{itemize}
    However, there are uncountably many other pairs $(\shape_1', \shape_2')$ in between the two extreme cases, if we allow both shapes to change. All of them can be captured by simultaneously satisfying $\gield^{}_{2',1'} = \gield^{}_{1',2'}= 0$.
\end{coro}

We can traverse the family of maximal collision-free pairs $(\shape_1', \shape_2')$ using a hyper-parameter $\gamma\in [0, 1]$ such that:
\begin{align}
    (\shape_1', \shape_2')_{\gamma} \big|_{\gamma := 0} = (\shape_1, \shape^\ast_2), \\
    (\shape_1', \shape_2')_{\gamma} \big|_{\gamma := 1} = (\shape^\ast_1, \shape_2), 
\end{align}
while $(\shape_1', \shape_2')_{\gamma}$ for $\gamma \in (0, 1)$ produces other maximal pairs in between. Although we cannot provide an explicit formula for $(\shape_1', \shape_2')_{\gamma}$, we present an optimization-based approach to compute the pair in Section \ref{sec_method}.

Note that the uncountable family of all maximal pairs cannot be parameterized with a single (or a number of) parameters. Our goal here is capture a specific subset of the family with a meaningful hyper-parameter---namely, an enforced volumetric ratio between the shapes of maximized volume---that serves as a ``knob'' to navigate other tradeoffs for design space exploration.

\subsection{Computational Discretization} \label{sec_disc} \label{sec_comp}

The above formulation is representation-agnostic and can be used with a variety of different representation schemes (e.g., B-reps, mesh, and voxels), as long as they support computation of $\dimm-$integrals---and if they do not, one can approximate $\dimm-$integrals by sampling the solids using quadrature rules, as long as the representation schemes support basic membership classification queries. Here, we select a simple asymmetric discretization approach to illustrate the practicality of our formulation.

To enable finite representation of solids and digital computation of collision measures, let us approximate the indicator function of a stationary solid $\shape \subseteq \R^\dimm$ using a finite volume (FV) scheme over a grid, meaning that each grid cell is associated with densities (i.e., $\dimm-$measure fractions) of the solid entrapped within the cell, denoted by $\rho^{}_{\shape,i} \in [0, 1]$ for $i = 1, 2, \ldots, n$. For a moving solid $\shape \subseteq \R^\dimm$, on the other hand, let us use a finite sample (FS) scheme over the same grid, meaning that each cell center is associated with the same density values like a ``lumped'' measure. A rigorous application of these discretization schemes to approximate the collision integrals by sums in a $\dimm-$measure-preserving fashion over any subset of the grid cell is presented in \ref{app_sec_disc}. The resulting sums are:
\begin{align}
    \gield^{}_{2,1} &\approx \sum_{i_1=1}^{n_1}\sum_{i_2=1}^{n_2} \rho^{}_{\shape_1,i_1}\rho^{}_{\shape_2,i_2} \weight^{1,2}_{i_1, i_2}, \label{eq_g1_finite} \\
    \gield^{}_{1,2} &\approx \sum_{i_1=1}^{n_1}\sum_{i_2=1}^{n_2} \rho^{}_{\shape_1,i_1}\rho^{}_{\shape_2,i_2} \weight^{2,1}_{i_2, i_1}, \label{eq_g2_finite} 
\end{align}
where the weights $\weight^{1,2}_{i_1, i_2}$ and $\weight^{2,1}_{i_2, i_1}$ mean the following:
\begin{itemize}
    \item $\weight^{1,2}_{i_1, i_2}$ measures how long the grid vertex $\bx^{}_{i_2} \in \Omega_2$ of the moving grid (attached to $\shape_2 \subseteq \Omega_2$) stays within the $\dimm-$cell $\cell_{i_1}$ under the motion $\tau_{1,2}(t)$ for $t \in [0, 1]$.
    \item $\weight^{2,1}_{i_2, i_1}$ measures how long the grid vertex $\bx^{}_{i_1} \in \Omega_1$ of the moving grid (attached to $\shape_1 \subseteq \Omega_1$) stays within the $\dimm-$cell $\cell_{i_2}$ under the motion $\tau_{2,1}(t)$ for $t \in [0, 1]$.
\end{itemize}
For further details on the FV and FS discretization schemes, see \ref{app_sec_comp}.
    
The finite approximations of the collision measures in \eq{eq_g1_finite} and \eq{eq_g2_finite} can be written as matrix equations:
\begin{align}
    \gield^{}_{2,1} &\approx \big[\rho^{}_{\shape_1,i_1} \big]^\mathrm{T} \big[ \weight^{1,2}_{i_1, i_2} \big] \big[ \rho^{}_{\shape_2,i_2} \big], \label{eq_g1_matrix} \\
    \gield^{}_{1,2} &\approx \big[\rho^{}_{\shape_2,i_2} \big]^\mathrm{T} \big[ \weight^{2,1}_{i_2, i_1} \big] \big[ \rho^{}_{\shape_1,i_1} \big], \label{eq_g2_matrix} 
\end{align}
The two arrays $[\rho^{}_{\shape_1,i_1}]_{n_1 \times 1}$ and $[\rho^{}_{\shape_2,i_2}]_{n_2 \times 1}$ are discrete representations of the two solids, i.e., the {\it design variables}. Importantly, the two matrices $[ \weight^{1,2}_{i_1, i_2} ]_{n_1 \times n_2}$ and $[ \weight^{2,1}_{i_2, i_1} ]_{n_2 \times n_1}$ do {\it not} depend on the designs. The weights can be viewed as {\it pairwise correlations} between primal grid nodes of a moving grid and dual grid cells of a stationary grid, which depend solely on the relative motion of the grids and the grid structure. %(in this case, fully characterized by the resolution $\epsilon > 0$). 
The matrices can thus be pre-computed offline and reused in iterative design of the two solids. This property is critical for computational tractability of iterative design optimization in Section \ref{sec_method}, given that computing the collision measures for arbitrarily complex shapes and motions can be prohibitive in an iterative loop running hundreds or thousands of iterations.

Figure \ref{fig_discrete} illustrates how the stationary and moving solids are discretized via FV and FS, respectively, and how the correlation matrix entries are computed in practice. 

\begin{figure}[ht]
	\centering
	\includegraphics[width=0.5\textwidth]{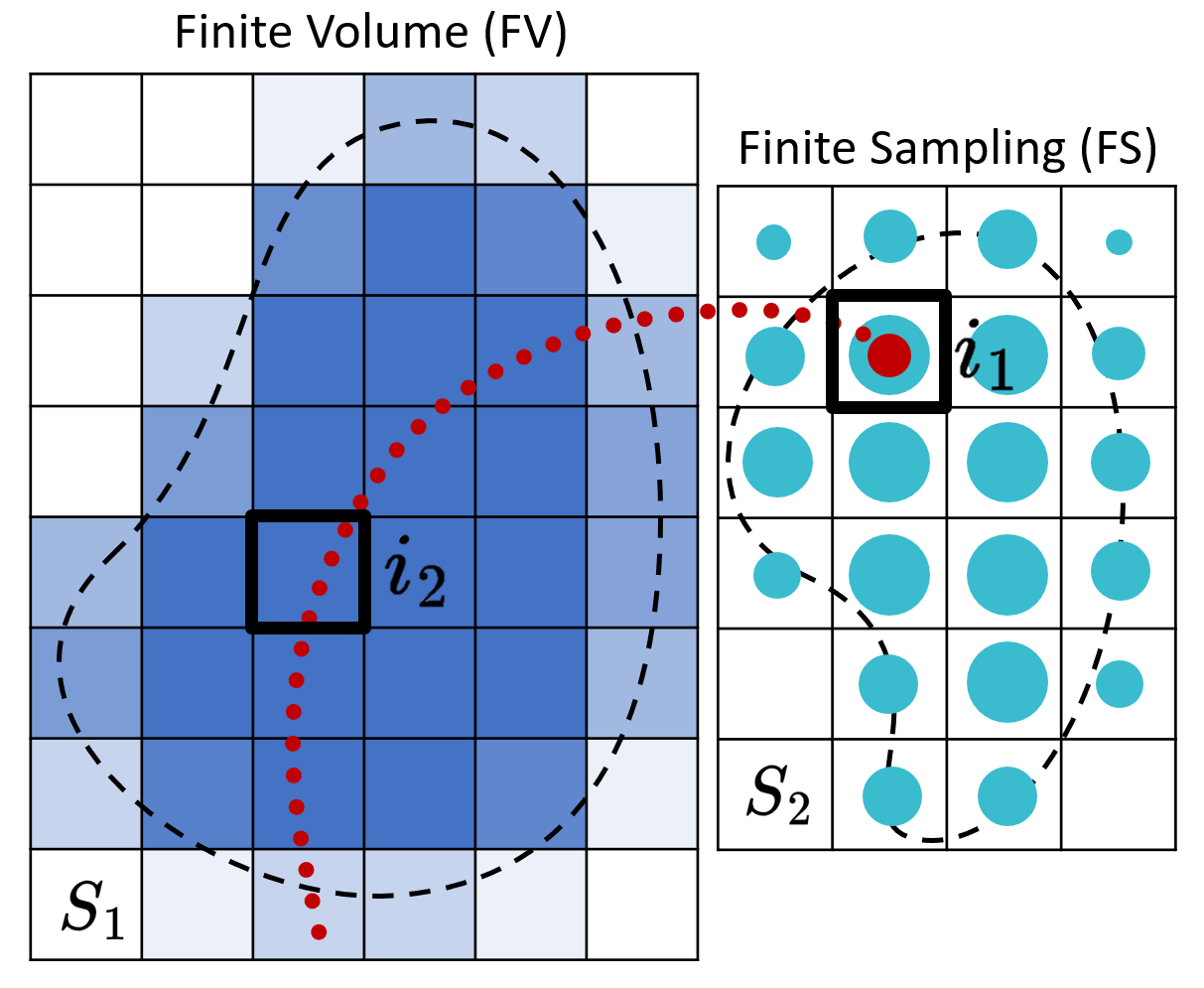}
	\caption{The stationary solid (in this case, $\shape_1$) is discretized via FV while the moving solid (in this case, $\shape_2$) is discretized via FS. The correlation matrix entry $\weight^{2,1}_{i_2, i_1}$ is computed by counting the number of times the discretized trajectory of the cell center at $i_1$ of the stationary grid stays within the cell at $i_2$ of the moving grid.}    
	\label{fig_discrete}
\end{figure}

To enable co-design of $[\rho^{}_{\shape_1,i_1}]_{n_1 \times 1}$ and $[\rho^{}_{\shape_2,i_2}]_{n_2 \times 1}$ subject to collision avoidance constraints, the violation of such constraints (i.e., the collision measures) must be differentiated with respect to the design variables. The resulting discrete sensitivity fields associated with the respective grids are computed using a chain rule:
\begin{align}
    \left[ \frac{\partial \gield^{}_{2,1}}{\partial \rho_{\shape_1, i_1}} \right] &\approx \big[ \weight^{1,2}_{i_1, i_2} \big] \big[ \rho^{}_{\shape_2,i_2} \big], ~
    \left[ \frac{\partial \gield^{}_{2,1}}{\partial \rho_{\shape_2, i_2}} \right] \approx \big[ \weight^{1,2}_{i_1, i_2} \big]^\mathrm{T} \big[ \rho^{}_{\shape_1,i_1} \big], \label{eq_g1_matrix_12} \\
    \left[ \frac{\partial\gield^{}_{1,2}}{\partial \rho_{\shape_1, i_1}} \right] &\approx \big[ \weight^{2,1}_{i_2, i_1} \big] \big[ \rho^{}_{\shape_1,i_1} \big], ~
    \left[ \frac{\partial\gield^{}_{1,2}}{\partial \rho_{\shape_2, i_2}} \right] \approx \big[ \weight^{2,1}_{i_2, i_1} \big]^\mathrm{T} \big[ \rho^{}_{\shape_2,i_2} \big], \label{eq_g2_matrix_21}
\end{align}
See \ref{app_diff_measure} for a definition of the topological sensitivity fields (TSFs), defined for the continuum (i.e., in terms of the indicator functions), corresponding to $\epsilon \to 0^+$.

\section{Generation of collision-free Geometries} \label{sec_method}
    
In this section, we present an optimization based formulation (Section \ref{sec_opt}) that allows both collision-free solids to evolve in a controlled manner, using the collision measures developed in the previous section.
    
\subsection{Optimization Problem Formulation} \label{sec_opt}

Using collision measures as constraints is not sufficient to define a well-posed optimization problem, because every subset of the maximal solids is also collision-free (despite not maximal). In practice, attempting to minimize the collision measures by gradient-descent optimization can lead to useless designs where the solids are excessively shrunk to avoid collision without maintaining contact. Figure \ref{fig_disks} shows an example in 2D; we begin with a stationary disk and a rotating disk around its center, and can satisfy collision avoidance with uncountably many solutions, adding/removing material from either disk. However, only two of them (I and II) are maximal pairs while (III) is not. 
Therefore, additional constraints are needed to achieve maximality.
    
\begin{figure}[ht]
   	\centering
   	\includegraphics[width=0.45\textwidth]{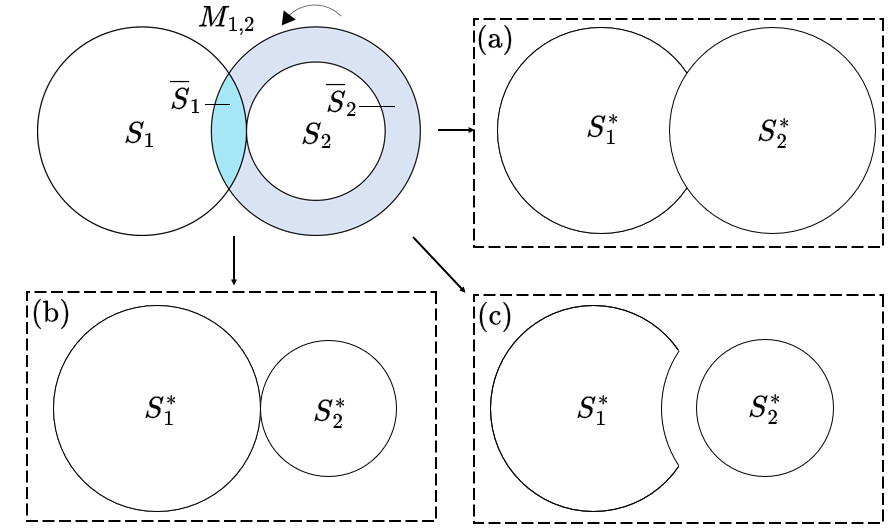}
   	\caption{Two initial designs $S_{1}$ and  $S_{2}$ are shown (top left), where $S_{2}$ moves relative to $S_{1}$ by a rotation about the center of $S_{2}$. Three collision-free pairs of solids are shown for this motion, where in (a) and (b) the pairs are maximal while in (c) they are not. Note that $S_{2}$ is the entire disk on the right including the white and teal regions.}    
   	\label{fig_disks}
\end{figure}
    
First, observe that we can partition the initial designs $\shape_1 \subseteq \Omega_1$ and $\shape_2 \subseteq \Omega_2$ into two subdomains each; namely, 
\begin{itemize}
	\item The {\it initially colliding} subdomains $\widehat{\shape}_1$ and $\widehat{\shape}_2$ are:
	\begin{align}
	\widehat{\shape}_1 \triangleq \big\{ \bx \in \shape_1 ~|~ \overline{\field}_{2,1}(\bx)>0 \big\}, \\
    \widehat{\shape}_2 \triangleq \big\{ \bx \in \shape_2 ~|~ \overline{\field}_{1,2}(\bx)>0 \big\},
	\end{align}
	\item The {\it initially collision-free} subdomains $\widetilde{\shape}_1$ and $\widetilde{\shape}_2$ are:
    \begin{align}
		\widetilde{\shape}_1 \triangleq \big\{ \bx \in \shape_1 ~|~ \overline{\field}_{2,1}(\bx)=0 \big\}, \\
        \widetilde{\shape}_2 \triangleq \big\{ \bx \in \shape_2 ~|~ \overline{\field}_{1,2}(\bx)=0 \big\}.
	\end{align}
\end{itemize}
Note that $\shape_1 = \widehat{\shape}_1 \cup \widetilde{\shape}_1$ and $\shape_2 = \widehat{\shape}_2 \cup \widetilde{\shape}_2$. We only need to eliminate the collisions by modifying the latter.

The simplest possible way to make the optimization problem well-posed is to pick an objective function that pushes the collision avoidance constraints to become {\it active}, i.e., the resulting pairs of solids become maximally collision-free. We observed that maximizing the $\dimm-$measure of both design subdomains (e.g., area in 2D and volume in 3D) is effective in many scenarios---although more sophisticated objective functions such as contact \cite{Lysenko2016effective} and complementarity \cite{Behandish2017shape} measures can be adopted in future work. Moreover, to control the choice of maximal pairs from a one-parametric family using a meaningful hyper-parameter $\gamma \in [0, 1]$, we add another constraint that enforces a $\dimm-$measure ratio between the two solids \eqref{generalForm_const3}. The optimization problem is thus formulated as finding $\widehat{\shape}_1 \subseteq \Omega_1$ and $\widehat{\shape_2} \subseteq \Omega_2$ to:
\begin{subequations} \label{eq_TOproblem}
	\begin{align}
		\mathop{\text{Maximize}}\limits_{\widehat{\shape}_1, \widehat{\shape_2}} \quad &
		f(\widehat{\shape}_1, \widehat{\shape}_2) = \mu^\dimm[\widehat{\shape}_1] + \mu^\dimm[\widehat{\shape}_2], \label{generalForm} \\
		\text{subject to} \quad 
		& {\gield}_{\widehat{2}, \widehat{1}} = 0, \label{generalForm_const1} \\
		& {\gield}_{\widehat{1}, \widehat{2}} = 0, \label{generalForm_const2} \\ 
		& g(\widehat{\shape}_1, \widehat{\shape}_2) = \gamma\mu^\dimm[\widehat{\shape}_1] - (1-\gamma)\mu^\dimm[\widehat{\shape}_2] = 0 \label{generalForm_const3} 
	\end{align}
\end{subequations}
For gradient-descent optimization, we define a Lagrangian from the general formulation:
\begin{equation}
	\mathcal{L}(\widehat{\shape}_1, \widehat{\shape}_2) = f(\widehat{\shape}_1, \widehat{\shape}_2) + \lambda_1\gield_{\widehat{2}, \widehat{1}} + \lambda_2\gield_{\widehat{1}, \widehat{2}} + \lambda_3 g(\widehat{\shape}_1, \widehat{\shape}_2).
	\label{generalLagrangian}
\end{equation}
Using the discretization scheme presented in Section \ref{sec_disc}, both $\widehat{\shape}_1 \subseteq \Omega_1$ and $\widehat{\shape_2} \subseteq \Omega_2$ can be represented by the density arrays $[\rho^{}_{\widehat{\shape}_1}]_{n_1 \times 1}$ and $[\rho^{}_{\widehat{\shape}_2}]_{n_2 \times 1}$, with pairwise collisions, $\big[\widehat{ \weight}^{1,2}_{i_1, i_2} \big]$ and $\big[\widehat{ \weight}^{2,1}_{i_2, i_1} \big]$ hence:
\begin{align}
	f (\widehat{\shape}_1, \widehat{\shape}_2) &\approx \epsilon^\dimm \big\| \rho^{}_{\widehat{\shape}_1}\big\|_1 + \epsilon^\dimm \big\| \rho^{}_{\widehat{\shape}_2}\big\|_1,\\
	g (\widehat{\shape}_1, \widehat{\shape}_2) &\approx \gamma \epsilon^\dimm \big\| \rho^{}_{\widehat{\shape}_1}\big\|_1 - (1-\gamma) \epsilon^\dimm \big\| \rho^{}_{\widehat{\shape}_2}\big\|_1,
\end{align}
where $\|\cdot\|_1$ is the $L_1-$norm, i.e., $\| \rho^{}_{\widehat{\shape}_1}\|_1$ and $\| \rho^{}_{\widehat{\shape}_2}\|_1$ are sum of non-negative density values in each array, approximating the $\dimm-$measure of the solids they represent:
\begin{align}
	\big\| \rho^{}_{\widehat{\shape}_1} \big\|_1 &\triangleq \sum_{i_1 = 1}^{n_1} \big| \rho^{}_{\widehat{\shape}_1, i_1} \big| = \sum_{i_1 = 1}^{n_1} \rho^{}_{\widehat{\shape}_1, i_1}, \\
	\big\| \rho^{}_{\widehat{\shape}_2} \big\|_1 &\triangleq \sum_{i_2 = 1}^{n_2} \big| \rho^{}_{\widehat{\shape}_2, i_2} \big| = \sum_{i_2 = 1}^{n_2} \rho^{}_{\widehat{\shape}_2, i_2}.
\end{align}
    
Without loss of generality, we can eliminate $\epsilon := 1$ from all equations, given that our approach is scale-agnostic. Substituting for the collision measures in \eq{eq_g1_matrix} and \eq{eq_g2_matrix} as well as the above norms in \eq{generalLagrangian}, we obtain:
\begin{align}
	\mathcal{L}(\widehat{\shape}_1, \widehat{\shape}_2) &\approx (1+\lambda_3 \gamma) \big\| \rho^{}_{\widehat{\shape}_1}\big\|_1 + (1-\lambda_3(1-\gamma)) \big\| \rho^{}_{\widehat{\shape}_2}\big\|_1 \nonumber \\
	&+ \lambda_1\big[\rho^{}_{\widehat{\shape}_1,i_1} \big]^\mathrm{T} \big[ \widehat{\weight}^{1,2}_{i_1, i_2} \big] \big[ \rho^{}_{\widehat{\shape}_2,i_2} \big] \nonumber \\
	&+ \lambda_2\big[\rho^{}_{\widehat{\shape}_2,i_2} \big]^\mathrm{T} \big[\widehat{\weight}^{2,1}_{i_2, i_1} \big] \big[ \rho^{}_{\widehat{\shape}_1,i_1} \big].
	\label{approxLagrangian}
\end{align}
All of the quantities in the above Lagrangian and its partial derivatives with respect to design variables can be readily computed through linear-algebraic operations, which scale well for parallel computing on CPU/GPU, although we have not implemented such parallelization in this paper. 
    
\subsection{Iterative Optimization Algorithm} \label{sec_alg}

To solve the optimization problem an iterative approach is utilized where the method of moving asymptotes (MMA) \cite{svanberg1987method} is utilized. Algorithm \ref{alg_TO} provides a description of the approach. 
		
\begin{algorithm} [ht!]
\caption{Co-generation of maximal collision-free solids}
	\begin{algorithmic}
		\Procedure{Co-generation}{$S_1, S_2, M_{1,2}, M_{2,1}, \gamma$}
		\State Initialize $[\rho_{\shape_1}] \gets [\indic_{S_1}]$
		\State Initialize $[\rho_{\shape_2}] \gets [\indic_{S_2}]$
		\State Initialize $\Delta \gets \infty$
		\State Initialize $iter \gets 0$
		\State $(\big[ \weight^{1,2}_{i_1, i_2} \big],\big[ \weight^{2,1}_{i_2, i_1} \big]) \gets \Call{PrecompCo}{\shape_1, \shape_2, M_{1,2}, M_{2,1}}$
		\State $([\widehat{\rho}],\big[ \widehat{\weight}^{1,2}_{i_1, i_2} \big],\big[\widehat{ \weight}^{2,1}_{i_2, i_1} \big])\hspace{-0.25em} \gets \hspace{-0.25em}\Call{ColVox}{[{\rho}],\big[ \weight^{1,2}_{i_1, i_2} \big],\big[ \weight^{2,1}_{i_2, i_1} \big]}$
		\While {$\Delta > \delta$ \textbf{and} $iter < l$ }
		\State $f \gets \Call{\text{Evaluate}}{[\widehat{\rho}]}$ \Comment{Obj. func.}
		\State $\frac{\partial \mathcal{L}_{[\widehat{\rho}]}}{\partial \widehat{\rho}} \gets
			\Call{\text{Gradient}}{[f,\widehat{\rho},\big[ \widehat{\weight}^{1,2}_{i_1, i_2} \big],\big[ \widehat{\weight}^{2,1}_{i_2, i_1} \big]}$\Comment{Sens.}
		\State $[\widehat{\rho}^{\text{new}}] \gets
			\Call{\text{Update}}{[\widehat{\rho}], \frac{\partial \mathcal{L}_{[\widehat{\rho}]}}{\partial \widehat{\rho}}} $ 
			\Comment{Gradient Update}
		\State $\Delta \gets \Call{Integrate}{[\widehat{\rho}^{\text{new}}]
			- [\widehat{\rho}]}$ \Comment{Vol. diff.}
		\State $iter \gets iter + 1$ \Comment{Iter.  counter}
		\State $[\widehat{\rho}] \gets [\widehat{\rho}^{\text{new}}]$
			\Comment{For next iteration}
		\EndWhile
		\State\Return{$[\widehat{\rho}]$}
		\EndProcedure 
	\end{algorithmic} \label{alg_TO}
\end{algorithm}

With the optimization formulated, we provide various examples of co-generation of collision-free solids.

\section{Results \& Discussion} \label{sec_results}

To demonstrate the efficacy of the approach, we provide a variety of 2D and 3D examples. All examples are run on a desktop machine with Intel\textsuperscript{\textregistered} CoreTM i7-9800X CPU with 16 processors running at 3.8 GHz, 32 GB of host memory, and an NVIDIA\textsuperscript{\textregistered} GeForce\textsuperscript{\textregistered} GTX 1080 GPU with 2,560 CUDA cores and 8 GB of device memory.

\subsection{2D Counter Rotating Squares}
The first example for the algorithm is two counter rotating squares where the angular position of the squares are provided by $\theta_1 \in [0, +2\pi)$ and $\theta_2 \in [0,-2\pi)$ (counter-clockwise) where the motion specified by $\theta_1 = -\theta_2 = \omega t$ with constant angular velocity $\omega > 0$. The dimensions and initial angular positions of the squares are shown in Fig. \ref{fig_rot_disks} (top). Additionally, the normalized local collision measures are provided in Fig. \ref{fig_rot_disks} (bottom) to show how initially the two shapes collide. The temporal resolution is $500$ time steps and the spatial resolution is $400 \times 400$ pixels.

\begin{figure}[h]
	\centering
	\includegraphics[width=0.5\textwidth]{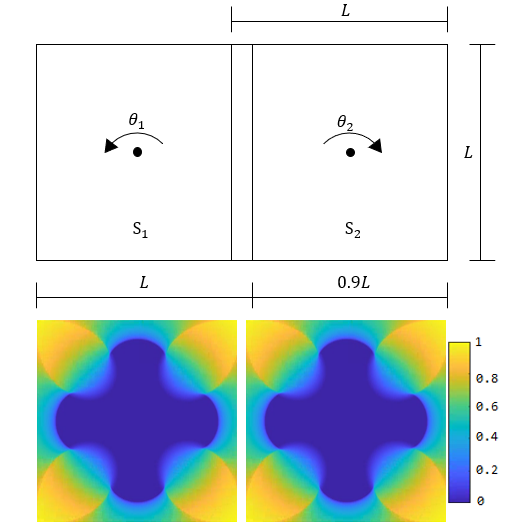}
	\caption{Two initial design (top), $S_{1}$ and  $S_{2}$, where each solid rotate around their centers with the same angular speed in opposite directions. The local collision measure for each domain is shown (bottom).} 
	\label{fig_rot_disks}
\end{figure}

The co-generation procedure was applied for $\gamma$ varying from 0 to 1 in 0.1 increments. To summarize the results, Fig. \ref{fig_gear_plot} provides a plot of the sum of the $\dimm-$measure of each co-generated solid as a function of $\gamma$. Selected solutions co-generated by the procedure are included as well. The solids are thresholded at $[{\rho}_{\shape_1,i_1}],[{\rho}_{\shape_2,i_2}] >0.5$.

\begin{figure} [h]
	\centering
	\includegraphics[width=0.5\textwidth]{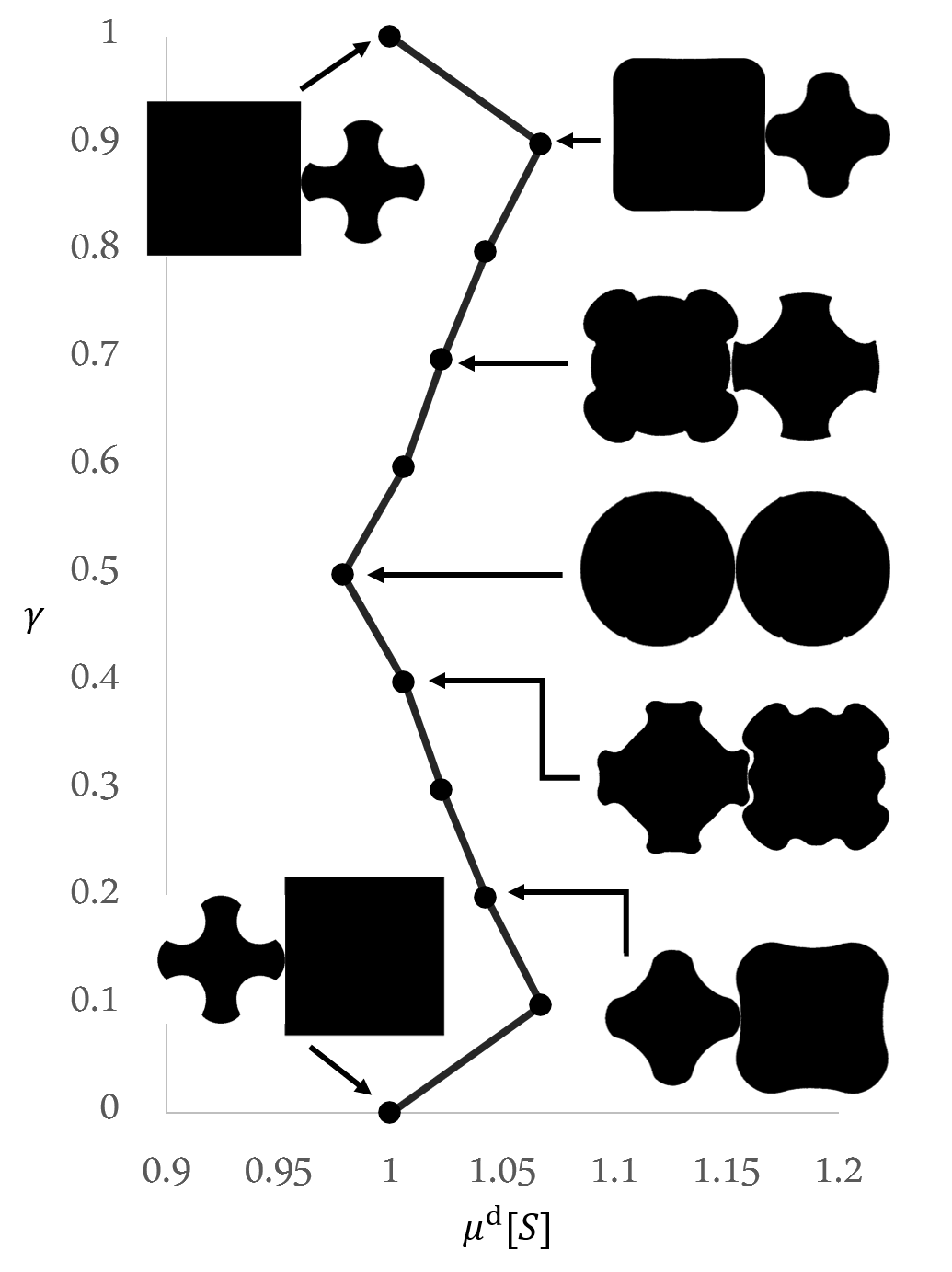}
	\caption{Sum of the $\dimm-$measures of the two collision-free solids as a function of hyper-parameter $\gamma$.} 
	\label{fig_gear_plot}
\end{figure}
    
Convergence plots are shown in Fig. \ref{fig_gear_conv} for $\gamma = 0.4$ and $\gamma = 0.5$ to demonstrate the convergence of the measure during incremental optimization.

\begin{figure}[h]
	\centering
	\includegraphics[width=0.5\textwidth]{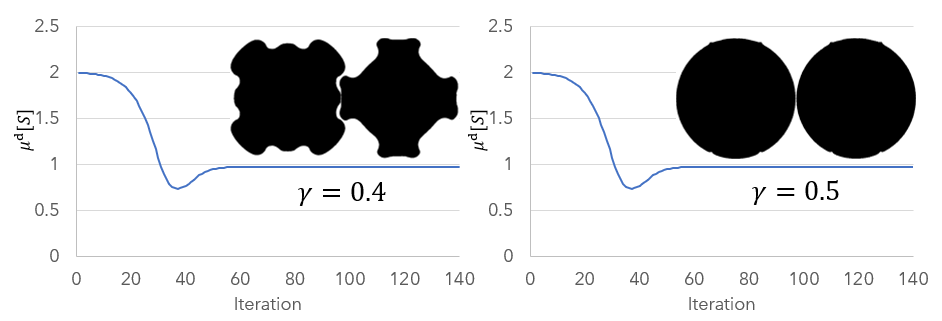}
	\caption{Convergence of summed $\dimm-$measures of the two collision-free solids for $\gamma = 0.4$ (left) and $\gamma = 0.5$ (right).} 
	\label{fig_gear_conv}
\end{figure}

Due to the symmetry of the trajectories and domains, the total weighted area as a function of $\gamma$ is symmetric about $\gamma = 0.5$. The solids also reflect this symmetry. Additionally, the results of $\gamma = 0$ and $\gamma = 1$ precisely match the results that would be obtained if the unsweep operations in \eq{eq_unsweep_1} and \eq{eq_unsweep_2} were performed. 

\subsection{2D Cam and Follower}

The second example is a cam/follower system. The cam, which is initially a square of length $L$, is prescribed to rotate $2\pi$ radians about its center (the origin $O$) at a constant angular velocity, while the follower translates in the vertical direction. The vertical position $y_{\mathrm{F}}$ of the center of the follower is prescribed as a function of the angular position $\theta_{C}$ of the cam:
\begin{equation}
    \label{followerEquation}
    y_{\mathrm{F}} = {{3L}\over{4}} + {{L}\over{8}}{\cos(2\theta_{\mathrm{C}})}.
\end{equation}
The dimensions and initial positions are shown in Fig. \ref{fig_cam2D}. Additionally in Fig. \ref{fig_cam2D} the normalized local collision measures are provided to show how the two shapes collide before incrementally changing their shapes. The plots  of local collision measures are offset and rescaled for clarity and do not correspond to the initial physical locations.

\begin{figure} [h]
	\centering
	\includegraphics[width=0.5\textwidth]{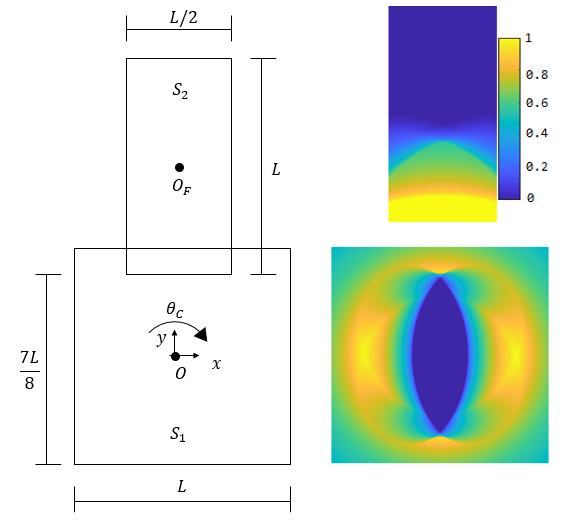}
	\caption{Two initial designs, cam ($\shape_\mathrm{C}$) and follower ($\shape_\mathrm{F}$) (left). The local collision measure for each domain is shown (right).} 
	\label{fig_cam2D}
\end{figure}

The temporal resolution is 1,000 time steps while the spatial resolutions are $400 \times 400$ for the cam and $400 \times 200$ for the follower. This problem is particularly interesting because the unsweep operations result in a loss of contact at certain times in the trajectory as shown in Fig. \ref{fig_loss_contact}. Our co-generation process produces persistent contact to ensure functionality, detailed below. 

\begin{figure} [h]
	\centering
	\includegraphics[width=0.5\textwidth]{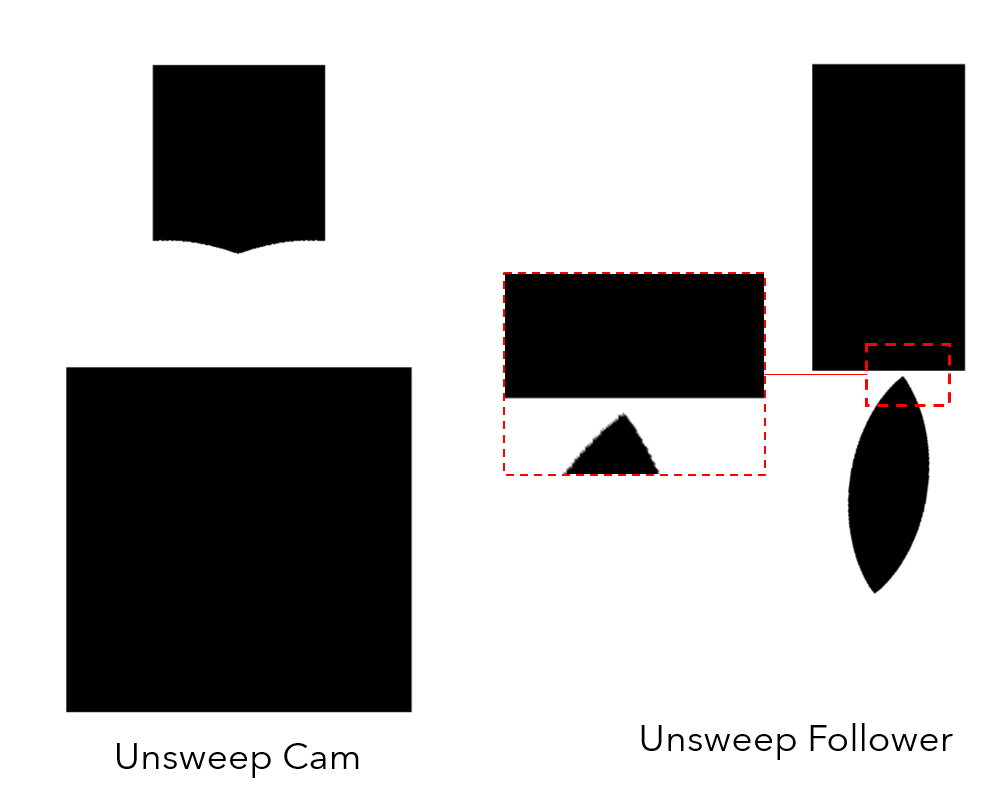}
	\caption{Illustration of how the unsweep operation can result in a loss of contact for the cam and follower.} 
	\label{fig_loss_contact}
\end{figure}

The co-generation procedure was applied for $\gamma$ varying from 0 to 1 in 0.1 increments. Fig. \ref{fig_cam_dist} provides a plot of the average minimum distance between the two solutions as they move through their trajectories as a function of $\gamma$.  As discussed, certain $\gamma$ values result in a loss of contact and therefore would not be functional. But in the range of approximately $\gamma \in (0.7,0.9)$ the solutions maintains contact even though it is not an explicit constraint. Remember that our approach does not guarantee persistent contact, as (a) it may not even be possible for any maximal pair for the given motion; (b) the proper shapes may not be obtainable by the specific subset of maximal pairs parameterized by $\gamma$ and computed by our incremental procedure. In particular, maximality in terms of volumes and collision avoidance may imply contact in one configuration, but not persistent contact throughout the whole motion, as shown in the results of Fig. \ref{fig_cam_dist}. Nevertheless, our observation is that co-generation is a necessary approach to create functional parts. Future work should focus on studying shape and motion properties with regards to persistent contact.  

\begin{figure} [h]
	\centering
	\includegraphics[width=0.5\textwidth]{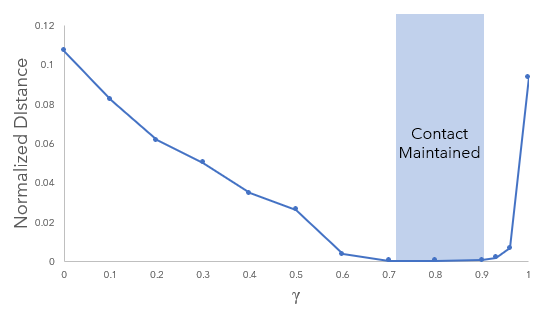}
	\caption{Average minimum distance between cam and follower as a function of $\gamma$ which illustrates co-generation is necessary to maintain contact during motion.} 
	\label{fig_cam_dist}
\end{figure}

To demonstrate the improved contact via co-generation, the solids generated when $\gamma = 0.8$ are shown in Fig. \ref{fig_cam2D_motion} at various positions during its trajectory. The cam and follower never lose contact during motion which is a necessary condition for functionality.

\begin{figure} [h]
	\centering
	\includegraphics[width=0.5\textwidth]{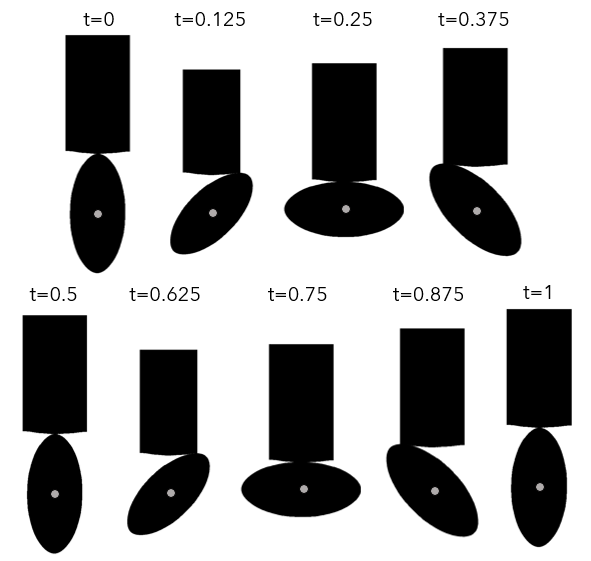}
	\caption{Cam and follower systems co-generated with $\gamma = 0.8$ shown at various configurations during the motion.} 
	\label{fig_cam2D_motion}
\end{figure}

The discretized sensitivity fields defined in \eq{eq_g1_matrix_12} and \eq{eq_g2_matrix_21} are plotted in Fig. \ref{fig_sensitivity} during the first iteration of the optimization loop. 

\begin{figure}[h]
	\centering
	\includegraphics[width=0.5\textwidth]{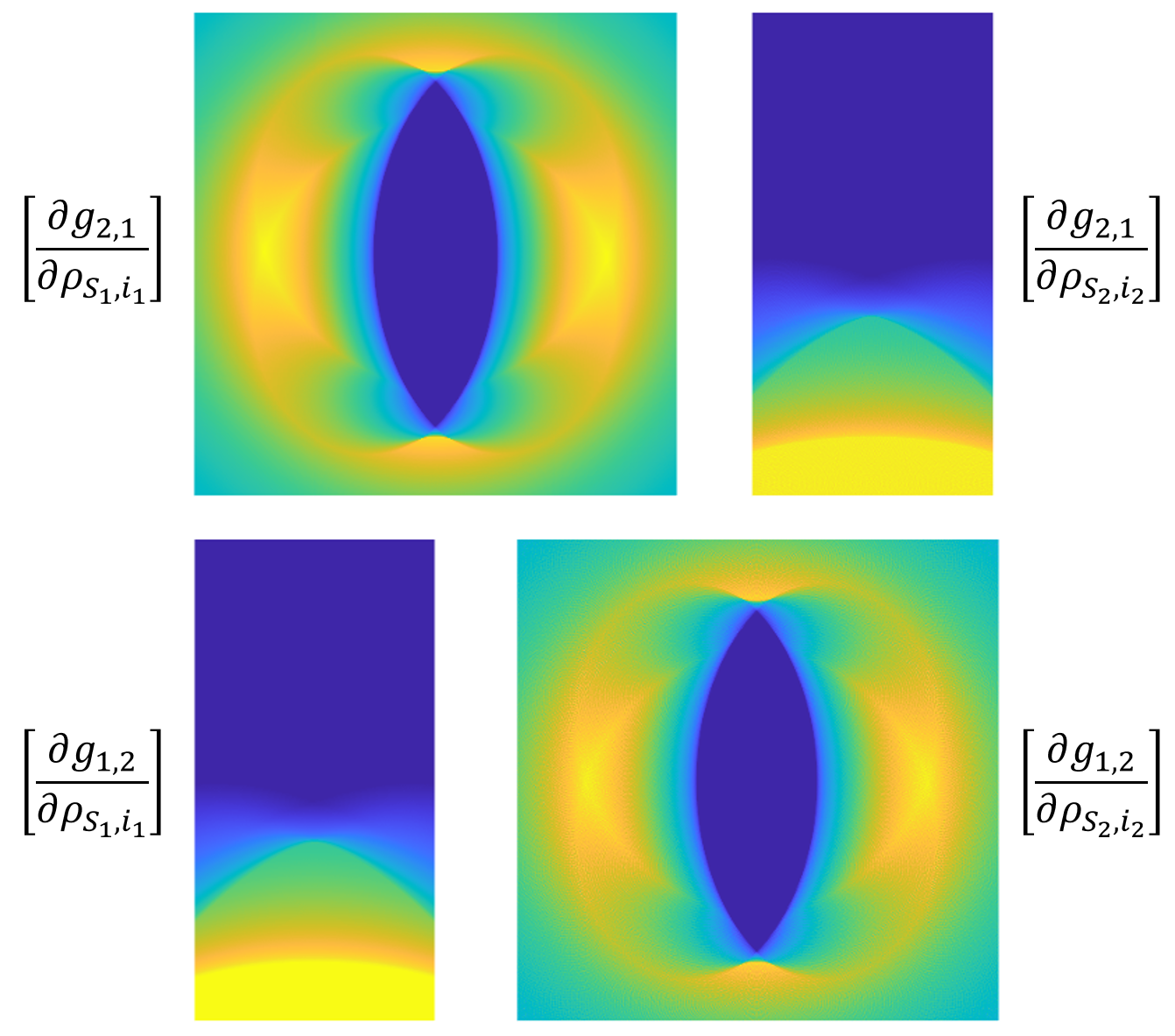}
	\caption{Plots showing the four discrete sensitivity fields for the initial 2D cam and follower example.} 
	\label{fig_sensitivity}
\end{figure}

\subsection{3D Cam and Follower}

The third example is a 3D cam/follower system. The cam rotates around the $z-$axis by an angle $\theta_\mathrm{C}$ at a constant angular velocity. Simultaneously, the follower rotates around the $x-$axis with by an angle $\theta_\mathrm{F}$ where $\theta_\mathrm{F} = \nicefrac{1}{2}|\sin(\theta_\mathrm{C})|$. The rotation matrices describing the motion of the cam and follower, $R_\mathrm{C}$ and $R_\mathrm{C}$, respectively, are provided below:
\begin{align}
    \label{RCam}
    R_{\mathrm{C}} &=  
    \begin{bmatrix}
    ~\cos(\theta_{\mathrm{C}}) & ~\sin(\theta_{\mathrm{C}}) & 0 \\
    -\sin(\theta_{\mathrm{C}}) & ~\cos(\theta_{\mathrm{C}}) & 0 \\
    0 & 0 & 1 \\
    \end{bmatrix}, \\
    \label{RFol}
    R_{\mathrm{F}} &=  
    \begin{bmatrix}
    1 & 0 & 0 \\
    0 & ~\cos(\theta_{\mathrm{F}}) & ~\sin(\theta_{\mathrm{F}}) \\
    0 & -\sin(\theta_{\mathrm{F}}) & ~\cos(\theta_{\mathrm{F}}) \\
    \end{bmatrix},
\end{align}
The design domain dimensions at initial configurations are shown in Fig. \ref{fig_cam3D}. The temporal resolution is 5,000 time steps, while the spatial resolutions are $105 \times 105 \times 105$ for the cam and $105 \times 105 \times 140$ for the follower. Each domain has over a million voxels. Similar to the 2D cam/follower system, if the unsweep operation was used to generate the collision-free shapes, a loss of contact would occur. By co-generating the solids, on the other hand, contact can be maintained as shown in Fig. \ref{fig_cam3D}, where the cam and follower are rendered in various relative configurations during the motion for $\gamma = 0.5$. In addition, a convergence plot is provided for $\gamma = 0.5$ in Fig. \ref{fig_cam3D_conv}.

\begin{figure} [h]
	\centering
	\includegraphics[width=0.5\textwidth]{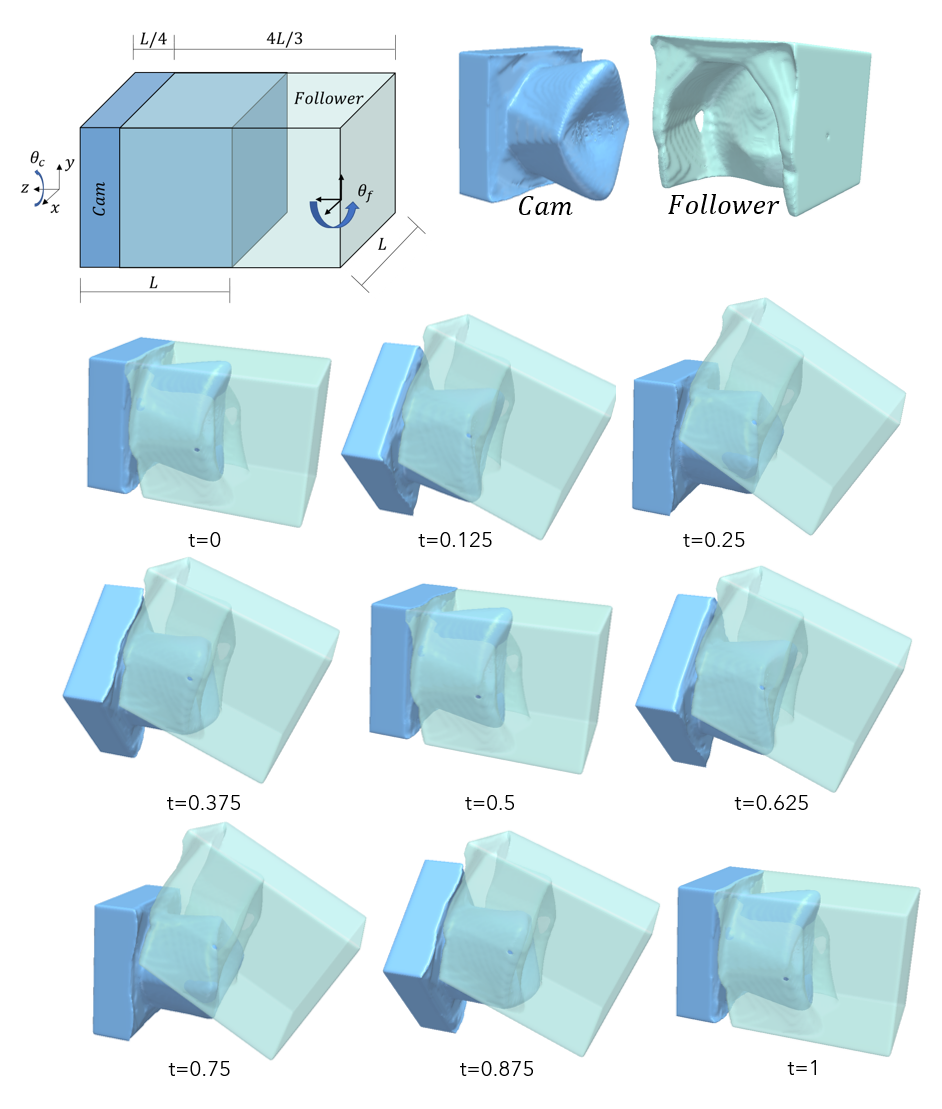}
	\caption{Initial solids for a spherical cam and follower.} 
	\label{fig_cam3D}
\end{figure}

\begin{figure} [h]
	\centering
	\includegraphics[width=0.5\textwidth]{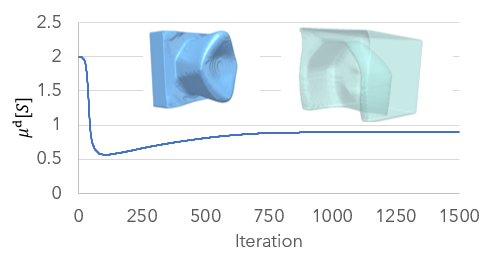}
	\caption{Convergence of summed $\dimm-$measures of the two collision-free solids for $\gamma = 0.5$.} 
	\label{fig_cam3D_conv}
\end{figure}

\subsection{3D Bolt and Nut}

The final example is a 3D bolt and nut pair. The dimensions and initial positions are shown in Fig. \ref{fig_bolt3D}, where $\phi_\mathrm{S}$ is the rotation angle of the bolt around the $z-$axis with a prescribed pitch of $P_\mathrm{S} = \nicefrac{L}{5}$, while the nut is stationary. The homogeneous transformation matrix $T_{\mathrm{S}}$, describing the relative screw motion is:
\begin{equation}
\label{screwEquation}
    R_{\mathrm{S}}=  
    \begin{bmatrix}
    ~\cos(\phi_{\mathrm{S}}) & ~\sin(\phi_{\mathrm{S}}) & 0 & 0 \\
    -\sin(\phi_{\mathrm{S}}) & ~\cos(\phi_{\mathrm{S}}) & 0  & 0\\
    0 & 0 & 1 & { {-L}\over{10\pi}}\phi \\
    0 & 0 & 0 & 1 \\
    \end{bmatrix}, \\
\end{equation}
where the translation ${{-L}\over{10\pi}}\phi_{\mathrm{S}}$ along and rotation $\phi_{\mathrm{S}} \in [0,8\pi)$ around the $z-$axis are linearly related. During the motion, the bolt makes 4 full turns and moves by $4 P_{\mathrm{S}} = 0.8 L$. The temporal resolution is 5,000 time steps, while the spatial resolutions are $50 \times 50 \times 150$ for the bolt and $100 \times 100 \times 50$ for the nut. One particular solution in which $\gamma = 0.2$ is provided in Fig. \ref{fig_bolt3D}.

\begin{figure} [!htb]
	\centering
	\includegraphics[width=0.5\textwidth]{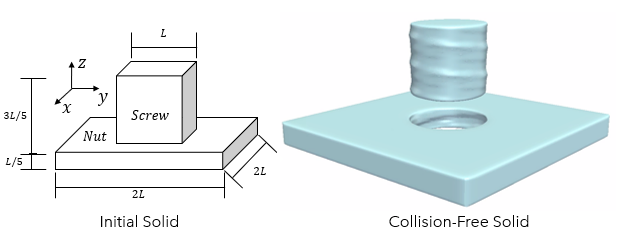}
	\caption{Initial design domains (left) and collision-free bolt and nut generated (right) with threads naturally emerging from co-design.} 
	\label{fig_bolt3D}
\end{figure}

The automatic co-generation of threads and convergence to a bolt and nut pair, starting from two colliding cubic regions, is a significant result. The threads appear naturally and automatically to generate maximal collision-free pairs for a given screw motion, leading to a bolt and nut pair with the same pitch as that of the given screw motion.
\begin{table*}  [!htb]
\centering
\caption{Summary of computation time for various examples provided in Section \ref{sec_results}.}
	\begin{tabular}{||l l l l l l||} 
		\hline
		Case Study & Spatial Res. (1) & Spatial Res. (2) & Temporal Res. & Mean Comp. Time (sec) & Max Iters. \\ [0.5ex] 
		\hline\hline
		2D Gear/Gear & 400$\times$400 & 400$\times$400 & 500 & 705.2 & 150 \\ 
		2D Cam/Follower & 400$\times$400 & 400$\times$400 & 1,000 & 949.1 & 200 \\
		3D Cam/Follower& 105$\times$105$\times$105 & 105$\times$105$\times$140 & 5,000 & 7,890 & 500\\
		3D Bolt/Nut & 50$\times$50$\times$50 & 100$\times$100$\times$50 & 5,000 & 4,201 & 350\\ [1ex] 
		\hline
	\end{tabular}
\label{table:1}
\end{table*}

However, note that the bolt and nut pair is not the only locally optimal solutions for a screw motion according to the formulation in \eq{eq_TOproblem}. A cylindrical peg and hole pair, for example, would satisfy all the enforced constraints, including maximal collision-avoidance, maximal (space-filling) volume, and desired volume ratio, while ensuring persistent contact as a side-effect. However, there are infinitely many screw and nut pairs (of different thread profiles) that could satisfy the constraints with the same maximal volume. Hence, it appears that finding one of these threaded solids is more likely than finding the degenerate extreme case of cylindrical peg and hole (i.e., no threads). It appears that our approach converged to one such local extremum of the optimization problem.

It is also worth noting that this procedure certainly does not generate all possible maximally collision-free pairs. Additional constraints or an entirely different set of objective functions and constraints---e.g., maximizing contact \cite{Lysenko2016effective} or complementarity \cite{Behandish2017shape} measures as opposed to total volume and volume ratio---may be used to identify more maximal pairs with better form, fit, or function.

Table \ref{table:1} summarizes the presented results. The spatial resolution of both design domains, the temporal resolution of the motions, the mean computation time for each optimization run, and the maximum number of optimization loop iterations are presented. Our numerical procedure required running times on the order of minutes and hours, indicating the scalability of the approach for high-resolution spatio-temporal discretizations. 

\subsection{Summary of Results \& Discussion}

By providing various examples, we demonstrated the ability of our approach to generate pairs of collision-free solids in both 2D and 3D, discretized by immersing them into a pair of Cartesian grids using FV and FS schemes. By precomputing the sparse pairwise correlation matrix for a pair of grids in relative motion, we demonstrated the scalability of the approach. The most computationally demanding task in the design loop is cast into a pre- and post-multiplication of the correlation matrix with the design variables (namely, the density arrays for each solid). The results indicate that co-generation is essential, not only for generating a broad family of collision-free solids to enable downstream design flexibility, but also to find well-fitting solids that maintain contact throughout the motion.

There are a few seemingly arbitrary choices that may affect the outcome of our approach. While the choice of initial solids can impact the final design in the extreme cases (e.g., when they are too small to collide resulting in an ill-posed optimization problem) the initial shapes appear less consequential than the motion and choice of the hyper-parameter (i.e., volume ratio) $\gamma$. There are uncountable pairs of continuum solids and exponentially many pairs of corresponding discrete (e.g., voxelized) solids that are maximal, among which we have chosen a one-dimensional subset parameterized by $\gamma$. A natural extension of the current work is to explore higher-dimensional subsets with more parameters, providing multiple ``knobs'' to designers to adjust additional/implicit design requirements while satisfying maximal collision-free contact.

Hopefully, our work will inspire formulating better objective functions and constraints that encode complex requirements such as physics, contact, and complementarity. A particularly promising direction is to combine the maximal collision avoidance with physics-based performance criteria (e.g., computed by FEA for a given set of loading conditions on the parts), manufacturability \cite{Mirzendehdel2020topology,Mirzendehdel2022topology}, and other constraints \cite{Morris2021topology} using the integration framework in \cite{Mirzendehdel2019exploring}.

\section{Conclusion} \label{sec_conclusion}

In this paper we introduced an automated and scalable procedure for generative co-design of pairs of solids under one-parametric motions to satisfy collision avoidance or (equivalently) containment constraints. Our automated procedure designs a {\it family} of ``maximal pairs'' of collision-free solids simultaneously. Our procedure can design shapes that maintain contact (when possible) while avoiding collisions, using a hyper-parameter to control the relative volume of their modified regions. Unlike existing methods that require making arbitrary restrictions such as parameterizing the shapes for common/simple motions or fixing one shape and computing the other via unsweep, our design process avoids premature decisions and unnecessary restrictions that lead to suboptimal solutions. We also show that unsweep of each solid against the other is subsumed as a special case of our operation at the two extremes of the hyper-parameter spectrum.

While collision avoidance is not sufficient to ensure persistent contact, our results suggest that incremental co-design to reduce collision while maximizing volume (or perhaps other objective functions) is a viable route for designing mechanisms in which contact and complementarity are critical for function, such as cams/followers and gear trains. More effective contact or complementarity measures can likely be employed, using similar measure-theoretic approaches used to quantify collision, to replace the simple volume maximization. Additionally, the formulation of the collision measures as differentiable fields enables seamless integration into gradient-based design optimization approaches such as TO. In the context of design for assembly, this could broaden the scope of TO and generative design from part-level to assembly-level design, in which parts and sub-assemblies are simultaneously evolved for physics-based performance, manufacturing, and other objectives while ensuring parts are collision-free. 

It should be noted that when co-designing solids within larger assemblies (e.g., 100s of parts), one has to deal with combinatorial explosion of collision analysis among pairs, triplets, and so on. Future work in this direction should explore efficient ways to avoid such exhaustive computations. In the worst case scenario, the initial design domain of each component can be selected largely enough for all parts in an assembly to collide with each other. Computing all of the collision matrices on top of simultaneously optimizing all parts would be computationally intractable. One possible mitigation is to come up with domain restriction strategies or spatial sorting/hashing to determine components that are collision-free without computing all motions. In addition to prohibitive running times, precomputing the correlation matrices will be memory-intensive. While the matrices are often sparse, it may not be the case for exceptionally high spatial resolutions and complex motions leading to large numbers of point-voxel collisions in each time-step. Finally the use of symmetries based on the motions should be explored to reduce computational time and memory costs.

\section*{Acknowledgments}

The authors are thankful to Amir Mirzendehdel for his insights on topology optimization and Saigopal Nelaturi for his support of this project under the Design for Assembly research program at PARC.

\section*{References}
\bibliography{cogenCollision}

\appendix
\appendix

\section{Shape Discretization Schemes} \label{app_sec_disc} \label{app_sec_comp}

To enable finite representation of solids and digital computation of collision measures, let us approximate the indicator function using a finite linear combination:
\begin{equation}
    \rho^{}_\shape(\bx) = \sum_{i = 1}^n \rho^{}_{\shape, i} \phi^{}_i(\bx), \quad\rho^{}_{\shape,i} \in [0, \infty),~ \phi^{}_i:\R^\dimm \to \R, \label{app_eq_rho}
\end{equation}
where $\phi^{}_i(\bx)$ are local basis functions associated with a spatial tessellation and the coefficients $\rho^{}_{\shape, i} \geq 0$ are the design variables representing an equivalence class of solids with the same $\dimm-$measure properties, for $i = 1, 2, \ldots, n$. To approximately compute the integrals in \eq{eq_subs_g1} and \eq{eq_subs_g2}, we consider an asymmetric discretization strategy in which we discretize the stationary solid (i.e., the one to which the frame of reference is attached) via a finite volume (FV) scheme, while we discretize the moving solid (i.e., the one whose motion is observed) via a finite sample (FS) scheme.

For both schemes, let us immerse the solid in a uniform Cartesian grid of edge length $\epsilon > 0$, which is small enough to resolve the minimum features of the solid. Let $\bx_i \in \R^\dimm$ be the coordinates of the i\th vertex (i.e., $0-$cell) on the grid. Let $\cell_i(\epsilon) \subset \R^\dimm$ stand for the dual $\dimm-$cell (e.g., square pixels in 2D and cubic voxels in 3D):
\begin{equation}
    \cell_i(\epsilon) \triangleq \big[\bx_i-\nicefrac{\epsilon}{2}, \bx_i+\nicefrac{\epsilon}{2} \big]^\dimm
\end{equation}
To capture the $\dimm-$integral quantities accurately, let us associate to the i\th grid vertex/cell the $\dimm-$measure of the solid confined to $\cell_i(\epsilon)$:
\begin{equation}
    \rho^{}_{\shape, i} \triangleq \mu^\dimm[\shape \cap C_i(\epsilon)] = \int_{C_i(\epsilon)} \indic_{\shape}(\bx) ~d\mu^\dimm[\bx]. \label{eq_density}
\end{equation}

For the FV scheme, the basis functions are selected as unit $\dimm-$measure pulse functions, i.e., indicator functions of the cells $\cell_i (\epsilon)$ normalized by the cell $\dimm-$measure $\mu^\dimm[C_i(\epsilon)] = \epsilon^\dimm$:
\begin{equation}
    \phi^{\fv}_i (\bx) \triangleq \frac{\indic_{\cell_i(\epsilon)}(\bx)}{\mu^\dimm[C_i(\epsilon)]} = \frac{1}{\epsilon^\dimm} \indic_{\cell(\epsilon)} (\bx - \bx_i),
\end{equation}
where $\cell(\epsilon) \triangleq [+\nicefrac{\epsilon}{2},-\nicefrac{\epsilon}{2}]$. Using these basis functions in the discrete formula in \eq{app_eq_rho} gives a piecewise-constant approximation of the indicator function over the dual grid.

For the FS scheme, on the other hand, the basis functions are selected as unit impulse functions, i.e., shifted Dirac delta functions that ``lump'' a unit $\dimm-$measure at each cell center:
\begin{equation}
    \phi^{\fs}_i (\bx) \triangleq \dirac^\dimm(\bx - \bx_i) \cong \lim_{\epsilon \to 0^+} \frac{1}{\epsilon^\dimm} \indic_{\cell(\epsilon)} (\bx - \bx_i).
\end{equation}
Note that $\phi^{\fs}_i (\bx) = \lim_{\epsilon\to0^+} \phi^{\fv}_i (\bx)$. Both basis functions have a unit $\dimm-$measure, by construction. Using these basis functions in the discrete formula in \eq{app_eq_rho} gives a lumped-measure approximation of the indicator function over the primal grid.

Note that both FV and FS schemes are constructed such that the $\dimm-$measures of the solid over every subset of the $\dimm-$cells (e.g., digitized areas in 2D and volumes in 3D) are captured exactly for the finite $\epsilon > 0$, while other inexact integral properties converge to their exact values as $\epsilon \to 0^+$.

Using an FV scheme to approximate the indicator functions of stationary solids $\shape_1 \subseteq \Omega_1$ and $\shape_2 \subseteq \Omega_2$, we obtain:
\begin{align}
    \rho^{\fv}_{\shape_1} (\bx) = \sum_{i_1=1}^{n_1} \rho^{}_{\shape_1,i_1} \phi^{\fv}_i (\bx) = \frac{1}{\epsilon^\dimm} \sum_{i_1=1}^{n_1} \rho^{}_{\shape_1,i_1} \indic_{\cell(\epsilon)}(\bx - \bx^{}_{i_1}), \label{app_eq_fv_1} \\
    \rho^{\fv}_{\shape_2} (\bx) = \sum_{i_2=1}^{n_2} \rho^{}_{\shape_2,i_2} \phi^{\fv}_i (\bx) = \frac{1}{\epsilon^\dimm} \sum_{i_1=1}^{n_1} \rho^{}_{\shape_2,i_2} \indic_{\cell(\epsilon)}(\bx - \bx^{}_{i_2}), \label{app_eq_fv_2} 
\end{align}
where $\rho^{}_{\shape_1,i_1}, \rho^{}_{\shape_2,i_2} \in [0, 1]$ are $\dimm-$measures of intersections $\shape_1 \cap \cell_{i_1}(\epsilon)$ and $\shape_2 \cap \cell_{i_2}(\epsilon)$, respectively., i.e., the portions of each solid trapped inside a given cell.

Using an FS scheme to approximate the indicator functions of relatively moving solids $\shape_{1,2}(t) = \tau_{2,1}(t)\shape_1$ and $\shape_{2,1}(t) = \tau_{1,2}(t)\shape_2$ defined in \eq{eq_S12} and \eq{eq_S21}, we obtain:
\begin{align}
    \rho^{\fs}_{\shape_{1,2}(t)} (\bx) &= \rho^{\fs}_{\shape_1} (\tau_{1,2}(t)\bx) \\
    &= \sum_{i_1=1}^{n_1} \rho^{}_{\shape_1,i_1} \dirac^\dimm \left(\tau_{1,2}(t)\bx - \bx^{}_{i_1} \right), \label{app_eq_fs_1} \\
    \rho^{\fs}_{\shape_{2,1}(t)} (\bx) &= \rho^{\fs}_{\shape_2} (\tau_{2,1}(t)\bx) \\
    &= \sum_{i_2=1}^{n_2} \rho^{}_{\shape_2,i_2} \dirac^\dimm \left(\tau_{2,1}(t)\bx - \bx^{}_{i_2} \right), \label{app_eq_fs_2} 
\end{align}
Substituting the indicator functions of the stationary and moving solids in \eq{eq_subs_g1} and \eq{eq_subs_g2} with their $\dimm-$measure-preserving approximations in \eq{app_eq_fv_1} through \eq{app_eq_fs_2}, respectively, and rearranging sums and integrals, we obtain:
\begin{align}
    \gield^{}_{2,1} &= \lim_{\epsilon\to0^+} \int_0^1 \int_{\Omega_1} \rho^{\fs}_{\shape_{2,1}(t)}(\bx) \rho^{\fv}_{\shape_1}(\bx) ~d\mu^\dimm[\bx]~dt \nonumber \\
    &= \lim_{\epsilon\to0^+} \sum_{i_1=1}^{n_1}\sum_{i_2=1}^{n_2} \rho^{}_{\shape_1,i_1}\rho^{}_{\shape_2,i_2} \weight^{1,2}_{i_1, i_2}, \label{app_eq_g1_finite} \\
    \gield^{}_{1,2} &= \lim_{\epsilon\to0^+} \int_0^1 \int_{\Omega_2} \rho^{\fs}_{\shape_{1,2}(t)}(\bx) \rho^{\fv}_{\shape_2}(\bx) ~d\mu^\dimm[\bx]~dt \nonumber \\
    &= \lim_{\epsilon\to0^+} \sum_{i_1=1}^{n_1}\sum_{i_2=1}^{n_2} \rho^{}_{\shape_1,i_1}\rho^{}_{\shape_2,i_2} \weight^{2,1}_{i_2, i_1}, \label{app_eq_g2_finite} 
\end{align}
where the weights $\weight^{1,2}_{i_1, i_2}$ and $\weight^{2,1}_{i_2, i_1}$ are defined by:
\begin{align}
    \weight^{1,2}_{i_1, i_2} &\triangleq \frac{1}{\epsilon^\dimm} \int_0^1 \int_{\Omega_1} \dirac^\dimm \left(\tau_{2,1}(t)\bx - \bx^{}_{i_2} \right) \indic_{\cell_{i_1}}(\bx) ~d\mu^\dimm[\bx]~dt, \nonumber \\
    \weight^{2,1}_{i_2, i_1} &\triangleq \frac{1}{\epsilon^\dimm} \int_0^1 \int_{\Omega_2} \dirac^\dimm \left(\tau_{1,2}(t)\bx - \bx^{}_{i_1} \right) \indic_{\cell_{i_2}}(\bx) ~d\mu^\dimm[\bx]~dt. \nonumber
\end{align}
The $\delta^\dimm-$function has a ``sifting'' property; meaning that it turns $\dimm-$integrals into a reading of the integrand at the impulse centers (where the $\delta^\dimm-$function's input is zero):
\begin{align}
    \weight^{1,2}_{i_1, i_2} &= \frac{1}{\epsilon^\dimm}  \int_0^1 \indic_{\cell_{i_1}} \left( \tau_{1,2}(t)\bx^{}_{i_2} \right) ~dt, \label{app_eq_w12} \\
    \weight^{2,1}_{i_2, i_1} &= \frac{1}{\epsilon^\dimm}  \int_0^1 \indic_{\cell_{i_2}} \left( \tau_{2,1}(t)\bx^{}_{i_1} \right) ~dt. \label{app_eq_w21}
\end{align}

To pre-compute the grid correlations (i.e., weights) in \eq{app_eq_w12} and \eq{app_eq_w21}, the time integral can be discretized using a simple Riemann sum approximation:
\begin{align}
    \weight^{1,2}_{i_1, i_2} &= \lim_{\delta\to0^+} \frac{\delta}{\epsilon^\dimm}  \sum_{k = 1}^K \indic_{\cell_{i_1}} \left( \tau_{1,2}(t_k)\bx^{}_{i_2} \right), \label{precomp_eq_w12} \\
    \weight^{2,1}_{i_2, i_1} &= \lim_{\delta\to0^+} \frac{\delta}{\epsilon^\dimm}  \sum_{k = 1}^K \indic_{\cell_{i_2}} \left( \tau_{2,1}(t_k)\bx^{}_{i_1} \right), \label{precomp_eq_w21}
\end{align}
where $t_k \triangleq (k+\nicefrac{1}{2})\delta$ for $k = 0, 1, \ldots, K-1$ to uniformly discretize the time period $[0, 1]$ with time-steps of $\delta = \nicefrac{1}{K}$, small enough to capture the minimum features of the motion trajectories. Computing each sum takes $O(K)$ operations to evaluate and apply the motions $\tau_{1,2}(t_k) \in \motion_{1,2}$ and $\tau_{2,1}(t_k) \in \motion_{2,1}$ to each of the grid vertices $\bx_{i_1} \in \Omega_1$ and $\bx_{i_2} \in \Omega_2$, and testing the membership of the displaced points against $\cell_{i_1}$ and $\cell_{i_2}$, respectively. For all time-steps and grid points, this takes $O(n_1 n_2 K)$ operations, which is carried out once in a pre-processing step and can be perfectly parallelized on the CPU/GPU.

\section{Collision Sensitivity Analysis} \label{app_diff_measure}

The additive/subtractive topological sensitivity fields (TSFs) (denoted by superscripts $+$/$-$, respectively) for the global collision measures $\gield^{}_{2,1}$ and $\gield^{}_{1,2}$ in \eq{eq_subs_g1} and \eq{eq_subs_g2}, respectively, are defined in terms of the changes in measures with respect to infinitesimal inclusion/exclusion in the shapes, for a given relative motion.
\begin{itemize}
	\item With respect to changes in the stationary shapes, the TSFs are defined by:
	\begin{align}
		\TS^+_1 \gield^{}_{2,1} (\bx) &\triangleq \lim_{\epsilon\to 0^+} \frac{\gield^{}_{2,1^+}(\bx) - \gield^{}_{2,1}}{\epsilon^\dimm}, \label{eq_TSF_1} \\
		\TS^-_1 \gield^{}_{2,1} (\bx) &\triangleq \lim_{\epsilon\to 0^+} \frac{\gield^{}_{2,1} - \gield^{}_{2,1^-}(\bx)}{\epsilon^\dimm}, \label{eq_TSF_2}
	\end{align}
	where $\gield^{}_{2,1^+}(\bx)$ and  $\gield^{}_{2,1^-}(\bx)$ are the global collision measures for the pairs $\big(\shape_1 \cup [\bx - \nicefrac{\epsilon}{2}, \bx + \nicefrac{\epsilon}{2}]^\dimm, \shape_2\big)$ and $\big(\shape_1 - (\bx - \nicefrac{\epsilon}{2}, \bx + \nicefrac{\epsilon}{2})^\dimm, \shape_2\big)$, respectively. Similarly:
	\begin{align*}
		\TS^+_2 \gield^{}_{1,2} (\bx) &\triangleq \lim_{\epsilon\to 0^+} \frac{\gield^{}_{1,2^+}(\bx) - \gield^{}_{1,2}}{\epsilon^\dimm}, \\
		\TS^-_2 \gield^{}_{1,2} (\bx) &\triangleq \lim_{\epsilon\to 0^+} \frac{\gield^{}_{1,2} - \gield^{}_{1,2^-}(\bx)}{\epsilon^\dimm},
	\end{align*}
	where $\gield^{}_{1,2^+}(\bx)$ and $\gield^{}_{1,2^-}(\bx)$ are the global collision measures for the pairs $\big(\shape_2 \cup [\bx - \nicefrac{\epsilon}{2}, \bx + \nicefrac{\epsilon}{2}]^\dimm, \shape_1\big)$ and $\big(\shape_2 - (\bx - \nicefrac{\epsilon}{2}, \bx + \nicefrac{\epsilon}{2})^\dimm, \shape_1\big)$, respectively.
	\item With respect to changes in the moving shapes, the TSFs are defined by:
	\begin{align*}
		\TS^+_2 \gield^{}_{2,1} (\bx) &\triangleq \lim_{\epsilon\to 0^+} \frac{\gield^{}_{2^+,1}(\bx) - \gield^{}_{2,1}}{\epsilon^\dimm}, \\
		\TS^-_2 \gield^{}_{2,1} (\bx) &\triangleq \lim_{\epsilon\to 0^+} \frac{\gield^{}_{2,1} - \gield^{}_{2^-,1}(\bx)}{\epsilon^\dimm},
	\end{align*}
	where $\gield^{}_{2^+,1}(\bx)$ and $\gield^{}_{2^-,1}(\bx)$ are the global collision measures for the pairs $\big(\shape_1, \shape_2 \cup [\bx - \nicefrac{\epsilon}{2}, \bx + \nicefrac{\epsilon}{2}]^\dimm\big)$ and $\big(\shape_1, \shape_2 - (\bx - \nicefrac{\epsilon}{2}, \bx + \nicefrac{\epsilon}{2})^\dimm\big)$, respectively. Similarly:
	\begin{align*}
		\TS^+_1 \gield^{}_{1,2} (\bx) &\triangleq \lim_{\epsilon\to 0^+} \frac{\gield^{}_{1^+,2}(\bx) - \gield^{}_{1,2}}{\epsilon^\dimm}, \\
		\TS^-_1 \gield^{}_{1,2} (\bx) &\triangleq \lim_{\epsilon\to 0^+} \frac{\gield^{}_{1,2} - \gield^{}_{1^-,2}(\bx)}{\epsilon^\dimm},
	\end{align*}
	where $\gield^{}_{1^+,2}(\bx)$ and $\gield^{}_{1^-,2}(\bx)$ are the global collision measures for the pairs $\big(\shape_2, \shape_1 \cup [\bx - \nicefrac{\epsilon}{2}, \bx + \nicefrac{\epsilon}{2}]^\dimm\big)$ and $\big(\shape_2, \shape_1 - (\bx - \nicefrac{\epsilon}{2}, \bx + \nicefrac{\epsilon}{2})^\dimm\big)$, respectively.
\end{itemize}

Hereafter, we develop the relationships for the first two TSFs $\TS^+_1 \gield^{}_{2,1} (\bx)$ and $\TS^-_1 \gield^{}_{2,1} (\bx)$ defined in \eq{eq_TSF_1} and \eq{eq_TSF_2}, respectively, for sensitivity analysis with respect to an infinitesimal inclusion/exclusion in the stationary solid $\shape_1$. The same relationships hold for the remaining six TSFs due to the symmetry in formulation with respect to index assignment and choice of frame of reference.

It is easy to verify that the additive/subtractive TSFs vanish when the query point is inside/outside, respectively, of the solid with respect to which the TSFs are computed, because the solid does not change due to union/inclusion with a small neighborhood:
\begin{align*}
	\exists \epsilon>0 ~\text{s.t.}~ \gield^{}_{2,1^+}(\bx) = \gield^{}_{2,1} &~\Rightarrow~ \TS^+_1 \gield^{}_{2,1} (\bx) = 0, \quad\text{if}~ \bx \in \interior\shape_1, \\
	\exists \epsilon>0 ~\text{s.t.}~ \gield^{}_{2,1^-}(\bx) = \gield^{}_{2,1} &~\Rightarrow~ \TS^-_1 \gield^{}_{2,1} (\bx) = 0, \quad\text{if}~ \bx \in \exterior\shape_1,
\end{align*}
noting that a solid $\shape_1$ is a {\it closed} regular set, i.e., contains both its interior (an open set) $\interior \shape_1 = \shape_1 - \partial \shape_1$ and boundary $\partial \shape_1$, while the exterior is the complement (also an open set): $\exterior\shape_1 = \overline{\shape}_1 = \R^\dimm - \shape_1$. The interior, boundary, and exterior partition the $\dimm-$space (in/on/out classification).

The collision integral in \eq{eq_subs_g1} is an {\it additive} property, i.e., the change in the integral due to the change in shape can be computed by applying the integral to the change itself:
\begin{align}
	\TS^+_1 \gield^{}_{2,1} (\bx) &= \lim_{\epsilon\to0^+} \frac{1}{\epsilon^\dimm} \int_0^1 \int_{C(\bx, \epsilon) - \shape_1} \indic_{\shape_{2,1}(t)}(\bx')~d\mu^\dimm[\bx']~dt, \label{eq_inc} \\
	\TS^-_1 \gield^{}_{2,1} (\bx) &= \lim_{\epsilon\to0^+} \frac{1}{\epsilon^\dimm} \int_0^1 \int_{C(\bx, \epsilon) \cap \shape_1} \indic_{\shape_{2,1}(t)}(\bx')~d\mu^\dimm[\bx']~dt, \label{eq_exc}
\end{align}
where $C(\bx, \epsilon) \triangleq (\bx - \nicefrac{\epsilon}{2}, \bx + \nicefrac{\epsilon}{2})^\dimm$.

For interior points $\bx_\interior \in \interior\shape_1$, we have $C(\bx_\interior, \epsilon) \subset \shape_1$ for small enough $\epsilon > 0$, hence $C(\bx_\interior, \epsilon) - \shape_1 = \emptyset$ and the integral in \eq{eq_inc} vanishes, as expected. But $C(\bx_\interior, \epsilon) \cap \shape_1 = C(\bx_\interior, \epsilon)$ and the integral in \eq{eq_exc} reduces to:
\begin{equation}
	\TS^-_1 \gield^{}_{2,1} (\bx_\interior) = \lim_{\epsilon\to0^+} \int_0^1 \int_{C(\bx_\interior, \epsilon)} \frac{1}{\epsilon^\dimm} \indic_{\shape_{2,1}(t)}(\bx')~d\mu^\dimm[\bx']~dt. \label{eq_exc_2}
\end{equation}
If, additionally, $\tau_{2,1}(t)\bx_\interior \in \interior\shape_2$, i.e., $\bx_\interior \in \tau_{1,2}(t)\interior\shape_2$, then $C(\bx_\interior, \epsilon) \subset \tau_{1,2}(t) \interior\shape_2 = \interior\shape_{2,1}(t)$ for small enough $\epsilon > 0$, hence the inner integral equals $\epsilon^\dimm/\epsilon^\dimm = 1$.

For exterior points $\bx_\exterior \in \exterior\shape_1$, we have $C(\bx_\exterior, \epsilon) \subset \shape_1$ for small enough $\epsilon > 0$, hence $C(\bx_\exterior, \epsilon) \cap \shape_1 = \emptyset$ and the integral in \eq{eq_exc} vanishes, as expected. But $C(\bx_\exterior, \epsilon) - \shape_1 = C(\bx_\exterior, \epsilon)$ and the integral in \eq{eq_inc} reduces to:
\begin{equation}
	\TS^+_1 \gield^{}_{2,1} (\bx_\interior) = \lim_{\epsilon\to0^+} \int_0^1 \int_{C(\bx_\exterior, \epsilon)} \frac{1}{\epsilon^\dimm} \indic_{\shape_{2,1}(t)}(\bx')~d\mu^\dimm[\bx']~dt. \label{eq_inc_2}
\end{equation}
If, additionally, $\tau_{2,1}(t)\bx_\exterior \in \interior\shape_2$, i.e., $\bx_\exterior \in \tau_{1,2}(t)\interior\shape_2$, then $C(\bx_\exterior, \epsilon) \subset \tau_{1,2}(t) \interior\shape_2 = \interior\shape_{2,1}(t)$ for small enough $\epsilon > 0$, hence the inner integral equals $\epsilon^\dimm/\epsilon^\dimm = 1$.

To summarize, if either $\bx \in \interior\shape_1$ or $\bx \in \exterior\shape_1$, then:
\begin{itemize}
	\item If $\tau_{2,1}(t)\bx \in \interior\shape_2$, then the inner integral is 1.
	\item If $\tau_{2,1}(t)\bx \in \exterior\shape_2$, then the inner integral is 0.
\end{itemize}
The case for query points that fall on the boundaries of either shape is slightly more complicated.

Assuming the curve $\tau_{2,1}(t)\bx$ intersects the boundary $\partial \shape_2$ over singular points (as opposed to continuous curve segments) the outer integral measures the duration of time over which the trajectory collides with the solid $\shape_2$ as $\epsilon \to 0^+$. Hence, the limit converges and both additive and subtractive TSFs exists.%
\footnote{If the curve remains tangent to the boundary over a continuous curve segment, the duration of time over which this happens is weighted depending on the local neighborhood shape in the limit $\epsilon \to 0^+$, but the fact that TSF is convergent does not change.}

Note also that the the formulas for additive and subtractive TSFs are identical, hence they converge to the same formula as $\|\bx_\exterior - \bx_\interior\|_2 \to 0^+$, meaning that the interior and exterior query points converge to a common meeting point $\bx \in \partial \shape_1$ on the boundary. However, the value of $\TS^+_1 \gield^{}_{2,1} (\bx)$ and $\TS^-_1 \gield^{}_{2,1} (\bx)$ when the query point is precisely on the boundary $\bx \in \partial \shape_1$ may be different, depending on the shape of the complementary partial neighborhoods $C(\bx, \epsilon) - \shape_1$ and $C(\bx, \epsilon) \cap \shape_1$. In other words, even though the TSFs exist and have limit relationships, they are discontinuous as the query point passes through the boundary.

Next, let us consider the discretization schemes presented in \ref{sec_disc} and \ref{app_sec_disc} and how the TSFs relate to the discrete derivatives given in \eq{eq_g1_matrix_12} and \eq{eq_g2_matrix_21} as $\epsilon \to 0^+$. Remember that the global collision measures can be approximated by \eq{eq_g1_finite} and \eq{eq_g2_finite} in the following sense:
\begin{align}
    \gield^{}_{2,1} &= \lim_{\epsilon\to0^+} \sum_{i_1=1}^{n_1}\sum_{i_2=1}^{n_2} \rho^{}_{\shape_1,i_1}\rho^{}_{\shape_2,i_2} \weight^{1,2}_{i_1, i_2}, \label{app_eq_g1_finite} \\
	\gield^{}_{1,2} &= \lim_{\epsilon\to0^+} \sum_{i_1=1}^{n_1}\sum_{i_2=1}^{n_2} \rho^{}_{\shape_1,i_1}\rho^{}_{\shape_2,i_2} \weight^{2,1}_{i_2, i_1}, \label{app_eq_g2_finite}
\end{align}
where $\rho^{}_{\shape_1,i_1}, \rho^{}_{\shape_2,i_2} \geq 0$ were defined by \eq{eq_density}. Applying the additive/subtractive TSF operators to both sides of \eq{app_eq_g1_finite}, which we explored above, yields:
\begin{equation}
	\TS^\pm_1 \gield^{}_{2,1} (\bx) = \lim_{\epsilon\to0^+} \sum_{i_1=1}^{n_1}\sum_{i_2=1}^{n_2} \left(\TS^\pm_1\rho^{}_{\shape_1,i_1} (\bx) \right)\rho^{}_{\shape_2,i_2} \weight^{1,2}_{i_1, i_2}, \label{eq_TSF_g_lim}
\end{equation}
noting that the infinitesimal inclusion/exclusion are applied to $\shape_1$ in this case, thus they affect $\rho^{}_{\shape_1,i_1}$ but not $\rho^{}_{\shape_2,i_2}$. 

Substituting for \eq{eq_density} and, once again, exploiting the additivity of the integral definition, we obtain:
\begin{align*}
	\TS^+_1 \rho^{}_{\shape_1, i_1}(\bx) &= \lim_{\epsilon'\to 0^+} \frac{1}{(\epsilon')^\dimm} \int_{C(\bx_{i_1}, \epsilon) - \shape_1} \indic_{C(\bx, \epsilon')}(\bx') ~d\mu^\dimm[\bx'], \\
	\TS^-_1 \rho^{}_{\shape_1, i_1}(\bx) &= \lim_{\epsilon'\to 0^+} \frac{1}{(\epsilon')^\dimm} \int_{C(\bx_{i_1}, \epsilon) \cap \shape_1} \indic_{C(\bx, \epsilon')}(\bx') ~d\mu^\dimm[\bx'].
\end{align*}
Once again, the following two cases are simple:
\begin{itemize}
	\item If the finite grid cell is completely inside $\shape_1$, then $C(\bx_{i_1}, \epsilon) \cap \shape_1 = C(\bx_{i_1}, \epsilon)$, hence:
	\begin{equation}
		\TS^-_1 \rho^{}_{\shape_1, i_1}(\bx) = \lim_{\epsilon'\to 0^+} \int_{C(\bx_{i_1}, \epsilon)} \frac{1}{(\epsilon')^\dimm}\indic_{C(\bx, \epsilon')}(\bx') ~d\mu^\dimm[\bx']. \label{eq_TSF_rho_exc}
	\end{equation}
	and $\TS^+_1 \rho^{}_{\shape_1, i_1}(\bx) = 0$ because $C(\bx_{i_1}, \epsilon) - \shape_1 = \emptyset$.
	\item If the finite grid cell is completely outside $\shape_1$, then $C(\bx_{i_1}, \epsilon) - \shape_1 = C(\bx_{i_1}, \epsilon)$, hence:
	\begin{equation}
		\TS^+_1 \rho^{}_{\shape_1, i_1}(\bx) = \lim_{\epsilon'\to 0^+} \int_{C(\bx_{i_1}, \epsilon)} \frac{1}{(\epsilon')^\dimm}\indic_{C(\bx, \epsilon')}(\bx') ~d\mu^\dimm[\bx']. \label{eq_TSF_rho_inc}
	\end{equation}
	and $\TS^-_1 \rho^{}_{\shape_1, i_1}(\bx) = 0$ because $C(\bx_{i_1}, \epsilon) \cap \shape_1 = \emptyset$.
\end{itemize}
Once again, the situation is more complicated when the grid cell cuts through the boundary, i.e., $C(\bx_{i_1}, \epsilon) \cap \partial \shape_1 \neq \emptyset$.

Note the subtle differences among the inclusion/exclusion center $\bx \in \R^\dimm$ (i.e., the point at which TSFs are queried), the grid cell center $\bx_{i_1} \in \R^\dimm$ (used to discretize $\shape_1$), the integration variable $\bx' \in \R^\dimm$, as well as the distinction between the grid cell size $\epsilon > 0$ (finite in \eq{eq_TSF_rho_exc} and \eq{eq_TSF_rho_inc} for a given grid, despite infinitesimal in \eq{eq_TSF_g_lim}) and the inclusion/exclusion size $\epsilon' > 0$ (infinitesimal).

Now, if we let $\bx = \bx_{i_1}$, for small enough $\epsilon' > 0$ we will have $\epsilon' < \epsilon$ thus $C(\bx, \epsilon') \subset C(\bx, \epsilon)$ and the identical integrals in \eq{eq_TSF_rho_exc} and \eq{eq_TSF_rho_exc} yield $(\epsilon')^\dimm/(\epsilon')^\dimm=1$, thus \eq{eq_TSF_g_lim} yields:
\begin{equation}
	\TS^\pm_1 \gield^{}_{2,1} (\bx_i) = \lim_{\epsilon\to0^+} \sum_{i_1=1}^{n_1}\sum_{i_2=1}^{n_2} \rho^{}_{\shape_2,i_2} \weight^{1,2}_{i_1, i_2}, \label{eq_TSF_g_simp}
\end{equation}
The right-hand side is the limit of the discrete derivative $\partial \gield^{}_{2,1}/\partial \rho_{\shape_1, i_1}$ in \eq{eq_g1_matrix_12}, proving that the discrete sensitivity converges to the topological sensitivity for the fully-internal and fully-external grid cells. If the cell center is not on the boundary, even if the cell cuts through the boundary, it will become fully-internal or fully-external for small enough $\epsilon > 0$ (i.e., resolution refinement), thus the convergence holds. If the cell center is precisely on the boundary, the derivation becomes slightly more involved, but follows a similar logic.

Note also that if two solids are in collision-free contact, the TSFs of the global collision measures queried at the boundary will be nonzero. The reason is that a small perturbation to the shape at the contact interface may result in collision. While this may create conservative material removal based on discrete sensitivity analysis (e.g., via TO), such issues are inevitable due to the loss of information as a result of discretization, e.g., ambiguity at the boundaries where pixels (2D) or voxels (3D) approximate an uncountably many different geometries with the same area (2D) or volume (3D). These issues can only be ameliorated by resolution refinement.

% \section*{References}

\end{document}